\documentclass[english,final,table]{IEEEtran}
\usepackage[T1]{fontenc}
\usepackage[latin9]{inputenc}
\usepackage{xcolor}
\usepackage{float}
\usepackage{mathrsfs}
\usepackage{bm}
\usepackage{amsmath}
\usepackage{amsthm}
\usepackage{amssymb}
\usepackage{graphicx}
\PassOptionsToPackage{normalem}{ulem}
\usepackage{ulem}

\makeatletter

\providecolor{lyxadded}{rgb}{0,0,1}
\providecolor{lyxdeleted}{rgb}{1,0,0}

\DeclareRobustCommand{\lyxsout}[1]{\ifx\\#1\else\sout{#1}\fi}

\theoremstyle{plain}
\newtheorem{lem}{\protect\lemmaname}
\theoremstyle{plain}
\newtheorem{prop}{\protect\propositionname}

\usepackage{amsmath,amssymb}
\usepackage[level=0]{wgroup_message}  

\usepackage{color}  
\usepackage{graphicx,psfrag,cite,subfigure}

\usepackage{flushend}

\author{
Wangqian Chen, Junting Chen and Shuguang Cui

\thanks{This work was supported in part by the National Science Foundation of China (NSFC) 
under Grant 62293482, in part by Basic Research Project under Grant HZQB-KCZYZ-2021067 of 
Hetao Shenzhen-HK S\&T Cooperation Zone, in part by NSFC Grant 62171398, 
in part by Shenzhen Science and Technology Program under Grant JCYJ20220530143804010, 
Grant KJZD20230923115104009, and Grant KQTD20200909114730003, 
in part by Guangdong Basic and Applied Basic Research Foundation 2024A1515011206, 
in part by Guangdong Research Grant 2019QN01X895, 
in part by the Guangdong Provincial Key Laboratory of Future Networks of Intelligence 
under Grant 2022B1212010001 and Guangdong-Hong Kong-Macao Joint Laboratory for Millimeter-Wave 
and Terahertz under Grant 2023B1212120002, in part by the National Key R\&D Program of China 
under Grant 2018YFB1800800, and in part by the Key Area R\&D Program of Guangdong Province 
under Grant 2018B030338001.
}

\thanks{W.~Chen, J.~Chen and S.~Cui are with the School of Science and Engineering, 
the Shenzhen Future Network of Intelligence Institute (FNii-Shenzhen), 
and the Guangdong Provincial Key Laboratory of Future Networks of Intelligence,
The Chinese University of Hong Kong, Shenzhen, Guangdong 518172, China 
(email: wangqianchen@link.cuhk.edu.cn; juntingc@cuhk.edu.cn; shuguangcui@cuhk.edu.cn).
}

}


\usepackage[acronym]{glossaries}
\newcommand{\newac}{\newacronym}
\newcommand{\ac}{\gls}
\newcommand{\Ac}{\Gls}
\newcommand{\acpl}{\glspl}

\makeglossaries
\newac{speb}{SPEB}{square position error bound}
\newac[plural=EFIMs,firstplural=Fisher information matrices (EFIMs)]{efim}{EFIM}{Fisher information matrix}
\newac{ne}{NE}{Nash equilibrium}
\newac{mse}{MSE}{mean squared error}
\newac{toa}{TOA}{time-of-arrival}
\newac{snr}{SNR}{signal-to-noise ratio}
\newac{lan}{LAN}{local area network}
\newac{psd}{PSD}{positive semidefinite}
\newac{pd}{PD}{positive definite}
\newac{wrt}{w.r.t.}{with respect to}
\newac{lhs}{L.H.S.}{left hand side}
\newac{wp1}{w.p.1}{with probability 1}
\newac{kkt}{KKT}{Karush-Kuhn-Tucker}
\newac{wlog}{w.l.o.g.}{without loss of generality}
\newac{mle}{MLE}{maximum likelihood estimation}
\newac{gps}{GPS}{global positioning system}
\newac{rssi}{RSSI}{received signal strength indicator}
\newac{mimo}{MIMO}{multiple-input multiple-output}
\newac{csi}{CSI}{channel state information}
\newac{fdd}{FDD}{frequency division duplexing}
\newac{ms}{MS}{mobile station}
\newac{bs}{BS}{base station}
\newac{d2d}{D2D}{device-to-device}
\newac{slnr}{SLNR}{signal-to-interference-leakage-and-noise-ratio}
\newac{ula}{ULA}{uniform linear antenna array}
\newac{pas}{PAS}{power angular spectrum}
\newac{mmse}{MMSE}{minimum mean square error}
\newac{zf}{ZF}{zero-forcing}
\newac{rzf}{RZF}{regularized zero-forcing}
\newac{as}{AS}{angular spread}
\newac{aod}{AOD}{angle of departure}
\newac{iid}{i.i.d.}{independent and identically distributed} 
\newac{sinr}{SINR}{signal-to-interference-and-noise ratio}
\newac{tdd}{TDD}{time-division duplex}
\newac{rvq}{RVQ}{random vector quantization}
\newac{rhs}{R.H.S.}{right hand side}
\newac{mrc}{MRC}{maximum ratio combining}
\newac{cdf}{CDF}{cumulative distribution function}
\newac{a.s.}{a.s.}{almost surely}
\newac{los}{LOS}{line-of-sight}
\newac{jsdm}{JSDM}{joint spatial division and multiplexing}
\newac{map}{MAP}{maximum a posteriori}
\newac{klt}{KLT}{Karhunen-Lo\`eve Transform}
\newac{lbe}{LBE}{link bargaining equilibrium}
\newac{se}{SE}{Stackelberg equilibrium}
\newac{uav}{UAV}{unmanned aerial vehicle}
\newac{nlos}{NLOS}{non-line-of-sight}
\newac{pdf}{PDF}{probability density function}
\newac{em}{EM}{expectation-maximization}
\newac{knn}{KNN}{$k$-nearest neighbor}
\newac{svd}{SVD}{singular value decomposition}
\newac{nmf}{NMF}{non-negative matrix factorization}
\newac{umf}{UMF}{unimodality-constrained matrix factorization}
\newac{rmse}{RMSE}{rooted mean squared error}
\newac{olos}{OLOS}{obstructed line-of-sight}
\newac{mmw}{mmW}{millimeter wave}
\newac{ber}{BER}{bit error rate}
\newac{rss}{RSS}{received signal strength}
\newac{lp}{LP}{linear program}
\newac{ufw}{U-FW}{unimodal Frank-Wolfe}
\newac{utf}{UTF}{unimodality-constrained tensor factorization}
\newac{fw}{FW}{Frank-Wolfe}
\newac{iot}{IoT}{Internet-of-Things}
\newac{mae}{MAE}{mean absolute error}
\newac{crb}{CRB}{Cram\'er-Rao bound}
\newac{aoa}{AoA}{angle of arrival}
\newac{wcl}{WCL}{weighted centroid localization}

\usepackage{enumitem}

\usepackage{algorithmic}

\usepackage{multirow}
\usepackage{makecell}
\usepackage{color}

\newac{dl}{DL}{deep learning} 
\newac{rt}{RT}{ray tracing} 
\newac{nn}{NN}{neural network} 
\newac{tx}{TX}{transmitter} 
\newac{rx}{RX}{receiver} 

\makeatother

\usepackage{babel}
\providecommand{\lemmaname}{Lemma}
\providecommand{\propositionname}{Proposition}

\begin{document}
\title{Physics-Informed Neural Networks for MIMO Beam Map and Environment
Reconstruction}
\maketitle
\begin{abstract}
As communication networks evolve towards greater complexity (e.g.,
6G and beyond), a deep understanding of the wireless environment becomes
increasingly crucial. When explicit knowledge of the environment is
unavailable, geometry-aware feature extraction from \ac{csi} emerges
as a pivotal methodology to bridge physical-layer measurements with
network intelligence. This paper proposes to explore the \ac{rss}
data, without explicit 3D environment knowledge, to jointly construct
the radio beam map and environmental geometry for a \ac{mimo} system.
Unlike existing methods that only learn blockage structures, we propose
an oriented virtual obstacle model that captures the geometric features
of both blockage and reflection. Reflective zones are formulated to
identify relevant reflected paths according to the geometry relation
of the environment. We derive an analytical expression for the reflective
zone and further analyze its geometric characteristics to develop
a reformulation that is more compatible with deep learning representations.
A physics-informed deep learning framework that incorporates the reflective-zone-based
geometry model is proposed to learn the blockage, reflection, and
scattering components, along with the beam pattern, which leverages
physics prior knowledge to enhance network transferability. Numerical
experiments demonstrate that, in addition to reconstructing the blockage
and reflection geometry, the proposed model can construct a more accurate
\ac{mimo} beam map with a 32\%$-$48\% accuracy improvement.
\end{abstract}


\section{Introduction}

\IEEEPARstart{M}{assive} \ac{mimo} transmission techniques have
enabled efficient spatial multiplexing, significant beamforming gain,
and flexible interference mitigation, which have found success in
5G networks and beyond. However, to achieve the full benefit of \ac{mimo}
systems, the high dimensional \ac{csi} is needed, which imposes significant
channel training overhead. In this context, \ac{mimo} beam maps are
constructed as charts or structured representations that capture the
spatial characteristics of the \ac{mimo} channel. Specifically, \ac{mimo}
beam maps may enable interference management and adaptive beamforming
without exhaustive channel training \cite{WeiXu:J22,FenLin:J24,ChiMas:J24},
and they have been exploited for optimizing beam alignment and signal
power allocation for energy efficiency enhancement \cite{XueJi:J24}.
However, while there have been some investigations on the exploitation
of \ac{mimo} beam maps, little is known in the literature on the
efficient construction of \ac{mimo} beam maps.

The fundamental challenge of \ac{mimo} beam map construction is the
demand of a large volume of measurement data. While there have been
some works on power spectrum map or radio environment map construction
using statistical methods \cite{HuZha:J20,DalRos:J22}, interpolation
\cite{PhiTonSic:C12,ZhaWan:J22}, tensor completion \cite{SunChe:J24,ZhaFu:J20}
and deep learning approaches \cite{KenSaiChe:J19,LevRonYap:J21,WuDanAi:20J},
these works mainly focused on constructing the \ac{rss} at each location
without considering the beam pattern of the antenna array. However,
constructing a \ac{mimo} beam map can be much more challenging than
constructing a power spectrum map. First, a \ac{mimo} beam map is
expected to capture higher dimensional \ac{csi} at each location
as opposed to conventional power maps that only capture the aggregate
power. Second, a \ac{mimo} beam map tends to experience larger spatial
variations, because the signal strength is affected {\em not only}
by the propagation environment, but also by the \ac{mimo} beam pattern.
Thus, following a conventional interpolation or tensor completion
type approach, more measurements will be needed. Third, the \ac{mimo}
beam map is more sensitive to the propagation environment, as it captures
received power across beam directions, which is much less homogeneous
as compared to conventional power maps.

\Ac{rt} methods have been adopted to construct \ac{mimo} beam maps
by simulating reflection, diffraction, and scattering \cite{HeAiKe:J19,MatMicShi:J20}.
While \ac{rt} provides detailed spatial and temporal information
about \ac{csi}, the cost of the detailed environmental geometry,
the computational complexity and the memory requirements are prohibitively
large. To accelerate computation, several studies leveraged the power
of neural networks to infer implicit \ac{csi} mappings directly from
raw \ac{rt} data \cite{AzpRaw:J14,RatChePaw:J21}. More end-to-end
neural network architectures, such as generative adversarial networks
\cite{ZhaWij:J23} or deep autoencoders \cite{TegRom:J21}, were proposed
for \ac{csi} learning. However, these methods still relied heavily
on detailed environmental information, such as precise city maps,
which limits their applicability in scenarios where the environmental
information is incomplete or unavailable. Channel charting was also
employed to learn low-dimensional charts of \ac{csi} while preserving
the local geometry of the channel \cite{kazHan:J23,DengTir:C21,StuMed:J18}.
However, recovering location labels for radio map construction remains
a non-trivial challenge due to the inherent complexity and ambiguity
in mapping \ac{csi} to physical coordinates.

To summarize, the existing deep learning approaches for \ac{mimo}
beam map construction lack a clear {\em physics model} and rely purely
on data, without explicitly exploiting the geometric relationship
between channel measurements and the environment. Consequently, it
is not clear whether the learned knowledge can be transferred to new
scenarios. Although it is possible to extend a power map to a \ac{mimo}
beam map with the \ac{mimo} steering vectors in \ac{los} scenarios,
predicting the directionality of beamforming gains in \ac{nlos} conditions
is challenging due to blockages and reflections. While classical methods
can construct \ac{mimo} beam maps independently for each individual
beam pattern, they may fail to capture inherent correlation among
beams, and thus the required amount of training data scales linearly
over the number of antennas, which is highly inefficient.

This paper attempts to build a deep learning model for \ac{mimo}
beam map construction embedded with explicit propagation model and
environment model. We propose to reconstruct the environmental geometry
as an intermediate step to assist for the \ac{mimo} beam map construction,
since all the beams of all \acpl{tx} share the same propagation environment.
Such a process resembles radio tomography \cite{HamBen:J14}, which
solves an inverse problem of reconstructing the environment from \ac{rss}
measurements. While our earlier attempts \cite{LiuChen:J23,WangqianChen:J24}
proposed virtual obstacles to capture the environmental geometry,
they are limited to reconstructing the blockage features, while reconstructing
the reflection characteristics remains a challenge. Furthermore, instead
of pursuing accuracy in environment reconstruction, our goal is to
learn a geometric representation of the environment shared among \acpl{tx}
for the best inference of the \ac{mimo} beam map.

Specifically, we address the challenges by leveraging blockage and
reflection geometry in neural network design to jointly construct
the \ac{mimo} beam map and virtual environment using only \ac{rss}
measurements. Towards this end, we propose to learn a set of {\em oriented virtual obstacles}
at possible locations, where each oriented virtual obstacle consists
of two attributes: height that determines whether a signal can be
blocked or reflected, and orientation that determines the direction
of signal reflection. By exploiting the geometric relationship between
paths and oriented virtual obstacles, a physics-informed model that
incorporates the beam pattern is proposed to capture the channel attenuation
of the direct path and reflected paths. To tackle the challenge of
solving an {\em inverse RT problem} that reconstructs oriented virtual
obstacles from \ac{rss} measurements, we address this complex inverse
problem by introducing the notion of {\em reflective zones} to identify
a subset of relevant propagation paths based on environmental geometry.
In addition, a model embedded neural network is developed to incorporate
the {\em reflective-zone-based geometry model} for a joint \ac{mimo}
beam map and oriented virtual obstacle construction. Our implementation
and experiments show that such a physics-informed design simplifies
the modeling complexity of the neural network, and shows a great transferability
to generate new \ac{mimo} beam maps in previously unseen scenarios.

\begin{figure}[!t]
\hspace{0.3cm}\includegraphics[scale=0.72]{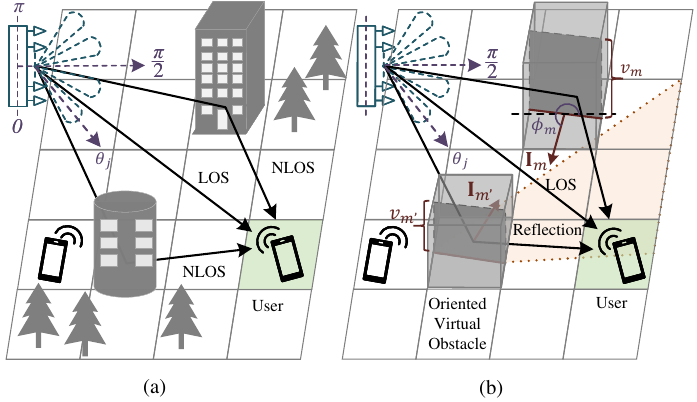}

\caption{a) A multipath propagation scenario with direct and single-bounce
reflection paths; b) Oriented virtual obstacles that describe the
geometry of the environment.}
\label{fig:Communication_Scene}
\end{figure}

The novelty and contribution are summarized as follows:
\begin{itemize}
\item We propose a physics-informed neural network architecture to {\em jointly}
construct the \ac{mimo} beam map and 3D environment from \ac{rss}
measurements. As the model explicitly learns the \ac{mimo} beam model,
the environment, and how the signal interacts with the environment,
it may reconstruct the \ac{mimo} beam maps much more efficiently.
\item We introduce the concept of reflective zones as deep learning representation
for signal reflection, and based on the geometry relation, we derive
simplified conditions that can be easily implemented in a deep learning
framework to characterize signal reflection.
\item Numerical experiments demonstrate that the proposed model achieves
over 30\% improvement in \ac{mimo} beam map accuracy compared to
existing deep learning based approaches, and exhibits superior extrapolation
capabilities to previously unseen scenarios.
\end{itemize}

In addition, we present an application of beam alignment where the
proposed method can reduce the search overhead by 78\% compared to
exhaustive beam sweeping.

The rest of the paper is organized as follows. Section II reviews
the system model, while Section III analyzes the reflective zone.
Section IV discusses the proposed learning framework. Design examples
are presented in Section V and conclusions are drawn in Section VI.

{\em Notations:} $(\cdot)^{\text{T}}$ and $(\cdot)^{\text{*}}$
represent the transpose and Hermitian transpose operations, respectively.
$\left|\cdot\right|$ denotes the absolute value and $\left\Vert \cdot\right\Vert $
represents the $L_{2}$ norm. $\overrightarrow{\boldsymbol{\mathbf{xy}}}$
denotes the vector from position $\mathbf{x}$ to position $\mathbf{y}$,
while $\mathbf{x\cdot y}$ represents the dot product of vectors $\mathbf{x}$
and $\mathbf{y}$, and therefore, $\mathbf{x\cdot y=x^{\text{T}}y}$.

\section{System Model}

Consider a massive \ac{mimo} communication system with multiple \acpl{bs}
in an outdoor urban environment. The \acpl{bs}, treated as \acpl{tx},
equip with multiple antennas, and all ground nodes, treated as \acpl{rx},
are single-antenna devices. For multi-antenna \acpl{rx}, they can
be considered as multiple co-located virtual single-antenna \acpl{rx}.
Denote a wireless link $\tilde{\mathbf{p}}=(\mathbf{p}_{\mathrm{t}},\mathbf{p}_{\mathrm{r}})\in\mathbb{R}^{6}$
using the position pair $\mathbf{p}_{\mathrm{t}},\mathbf{p}_{\mathrm{r}}\in\mathbb{R}^{3}$
of the \ac{tx} and the \ac{rx}, respectively.

The narrow-band channel of $\tilde{\mathbf{p}}$ can be expressed
as
\begin{equation}
\mathbf{h}(\tilde{\mathbf{p}})=\sum_{l=0}^{L}\beta_{l}\mathbf{a}^{*}(\phi_{l}(\tilde{\mathbf{p}}))\label{eq:channel_vector}
\end{equation}
where $\beta_{l}\sim\mathcal{CN}(0,\sigma_{l}^{2})$ is the complex
gain of the $l$th propagation path with zero mean and variance $\sigma_{l}^{2}$,
which captures the path loss, and $\mathbf{a}(\phi_{l})$ denotes
the array response vector for \ac{aod} $\phi_{l}$. For example,
in case of a \ac{ula}, the array response vector can be formulated
as
\begin{equation}
\mathbf{a}(\phi_{l})=\frac{1}{\sqrt{N_{t}}}[1,e^{-i2\pi\frac{\Delta}{\lambda}\textrm{cos}\phi_{l}},...,e^{-i2\pi\frac{\Delta}{\lambda}(N_{t}-1)\textrm{cos}\phi_{l}}]^{\textrm{T}}\label{eq:steering_vector}
\end{equation}
where $i=\sqrt{-1}$ is the imaginary unit, $\lambda$ is the carrier
wavelength, and $\Delta$ stands for the antenna spacing. 

A codebook-based approach at the \ac{bs} is applied to construct
\ac{mimo} beams for channel measurement. Denote $\mathbf{w}_{j}$
as the $j$th beamforming vector from a codebook $\mathcal{W}$ that
satisfies the power constrain with $|\mathbf{w}_{j}^{\ast}\mathbf{w}_{j}|=1$.
The beam pattern over a range of angle $\phi$ of the $j$th beam
is given by
\begin{equation}
B(\phi,\mathbf{w}_{j})=|\mathbf{a}^{*}(\phi)\mathbf{w}_{j}|^{2}.\label{eq:beam_pattern}
\end{equation}

It is assumed that the beamforming vectors $\mathbf{w}_{j}$ are designed
according to the antenna array geometry such that the beam pattern
$B(\phi,\mathbf{w}_{j})$ concentrates the energy in one specific
main-lobe direction $\theta_{j}$ and the beamforming gain decreases
as the angle deviates from the direction $\theta_{j}$ as shown in
Fig.~\ref{fig:Beam_Pattern}(a). It is also assumed that the number
of antennas $N_{t}$ at the BS is large, such that the sidelobe energy
is negligible compared to the measurement noise. Moreover, the collection
$\mathcal{W}$ of the beamforming vectors $\mathbf{w}_{j}$ are carefully
designed such that these beams cover the entire angular space as illustrated
in Fig.~\ref{fig:Beam_Pattern}(b).

An example under \ac{ula} is to adopt the discrete Fourier transform
(DFT) codebook. It is known that in the \ac{los} case under a large
$N_{t}$, the DFT matrix approximately diagonalizes the channel covariance
matrix, and hence, using DFT beams may obtain a good characterization
of the \ac{mimo} channel.

The measurement for the channel of $\tilde{\mathbf{p}}$ using beamforming
vector $\mathbf{w}_{j}$ is given by 
\begin{equation}
\rho(\tilde{\mathbf{p}},\mathbf{w}_{j})=\mathbb{E}\left[\left|\mathbf{h}(\tilde{\mathbf{p}})\mathbf{w}_{j}\right|^{2}\right]+n\label{eq:channel_gain_model}
\end{equation}
where $\mathbb{E}[\cdot]$ denotes the expectation over the randomness
from the small-scale fading, and $n$ denotes the measurement noise
and the uncertainty due to the small-scale fading.

For each location pair $\tilde{\mathbf{p}}$, a set of measurements
$\rho(\tilde{\mathbf{p}},\mathbf{w}_{j})$ are taken for $j=1,2,\dots,|\mathcal{W}|$,
where $|\mathcal{W}|$ denotes the number of beamforming vectors in
the codebook $\mathcal{W}$. The goal of the paper is to build a model
to generate \ac{mimo} beam maps $g(\tilde{\mathbf{p}},\mathbf{w}_{j})$
for all possible location pairs $\tilde{\mathbf{p}}$ over all beams
$\mathbf{w}_{j}\in\mathcal{W}$ from a set of limited measurements
\{$\rho(\tilde{\mathbf{p}},\mathbf{w}_{j}),\tilde{\mathbf{p}}$\}.

\begin{figure}[!t]
\centering\includegraphics[scale=0.32]{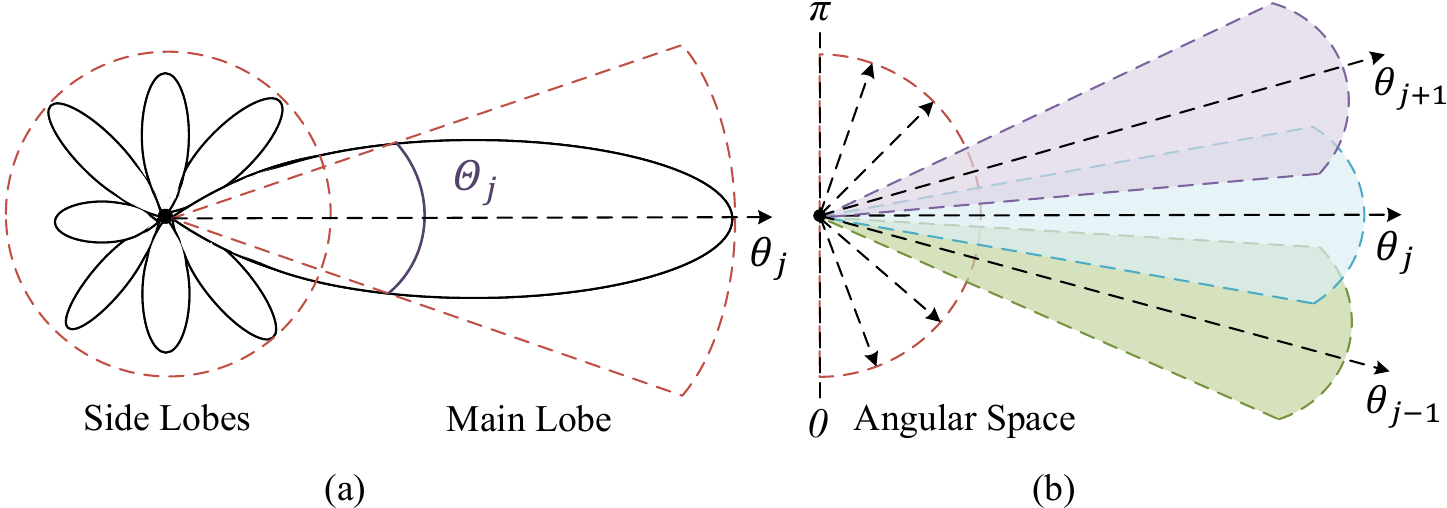}

\caption{a) The beam pattern studied consists of a main lobe and several sidelobes;
b) The beams are designed to collectively cover the angular space.}
\label{fig:Beam_Pattern}
\end{figure}

\subsection{Environment-Embedded Radio Map Model}

We adopt an oriented virtual obstacle model to describe the propagation
environment as follows. Partition the ground area of interest into
$M$ grid cells and denote $D$ as the grid spacing. The oriented
virtual obstacle on the $m$th grid cell is modeled as an oriented
cube of a certain height perpendicular to the ground, and the cube
may block the signal or reflect it to a specific direction based on
its orientation, as shown in Fig.~\ref{fig:Communication_Scene}.

Specifically, denote $v_{m}$ as the height of the virtual obstacle
and $\phi_{m}\in[0,2\pi)$ as the normal angle of its orientation
on the $m$th grid cell, respectively. The normal angle $\phi_{m}$
can also be captured by its normal vector $\mathbf{l}_{m}$, where
$\mathbf{l}_{m}=(\cos\phi_{m},\sin\phi_{m},\eta\text{)}$ and $\eta$
is a constant. Denote $\mathbf{V}_{\textrm{H}}$ as the matrix form
of $v_{m}$ such that the $(x,y)$th entry of $\mathbf{V}_{\textrm{H}}$
denotes the height at the $(x,y)$th grid cell. Similar notation $\mathbf{V_{\mathrm{\Phi}}}$
is defined for the normal angle $\phi_{m}$. The environmental geometry
is fully captured by the concatenated variable $\mathbf{V}=[\mathbf{V}_{\textrm{H}},\mathbf{V_{\mathrm{\Phi}}}]$.

Denote $\varphi(\boldsymbol{x},\boldsymbol{y})$ as the direction
angle of location $\boldsymbol{x}$ relative to location $\boldsymbol{y}$,
and $\mathbf{c}_{m}$ as the coordinates of the $m$th grid cell.
We formulate the first term in the measurement model (\ref{eq:channel_gain_model})
using an environment-embedded radio map model as
\begin{equation}
\begin{aligned}g(\tilde{\mathbf{p}},\mathbf{w}_{j};\mathbf{V}) & =B(\varphi(\mathbf{p}_{\mathrm{t}},\mathbf{p}_{\mathrm{r}}),\mathbf{w}_{j})g_{\mathrm{0}}(\tilde{\mathbf{p}},\mathbf{V})\\
 & \quad+\sum_{m}B(\varphi(\mathbf{p}_{\mathrm{t}},\mathbf{c}_{m}),\mathbf{w}_{j})g_{\mathrm{r}}^{m}(\tilde{\mathbf{p}},\mathbf{w}_{j},\mathbf{V})\\
 & \quad+g_{\mathrm{s}}(\tilde{\mathbf{p}},\mathbf{w}_{j},\mathbf{V})
\end{aligned}
\label{eq:radio-map-model}
\end{equation}
where the three terms characterize the following phenomenon:
\begin{itemize}
\item Direct path: In the first term, the function $B(\varphi(\mathbf{p}_{\mathrm{t}},\mathbf{p}_{\mathrm{r}}),\cdot)$
captures the beam pattern for the direct path from $\mathbf{p}_{\mathrm{t}}$
to $\mathbf{p}_{\mathrm{r}}$, and the function $g_{0}(\cdot)$ models
the path gain of the direct path based on the geometry of the environment
$\mathbf{V}$.
\item Reflected paths: All virtual obstacles have the potential to reflect
the signals from the \ac{tx}, and some of them reach the \ac{rx}.
Thus, the second term of (\ref{eq:radio-map-model}) aggregates contributions
from all possible reflected paths. The function $B(\varphi(\mathbf{p}_{\mathrm{t}},\mathbf{c}_{m}),\cdot)$
captures the beam pattern of the reflected path bounced by the $m$th
virtual obstacle, and the function $g_{\mathrm{r}}^{m}(\cdot)$ models
the reflected path gain.
\item Scattering paths: The third term $g_{\mathrm{s}}(\cdot)$ accounts
for the residual fluctuation due to scattering.
\end{itemize}

We further focus on designing simple representations for each component
to enable efficient learning of the model (\ref{eq:radio-map-model}).

\subsection{Blockage Model}

We establish a simplified representation for the relationship between
the direct path gain $g_{\mathrm{0}}(\tilde{\mathbf{p}},\mathbf{V})$
and the environment $\mathbf{V}$.

Denote $\mathcal{B}(\tilde{\mathbf{p}})$ as the set of grid cells
covered by the line segment joining $\mathbf{p}_{\mathrm{t}}$ and
$\mathbf{p}_{\mathrm{r}}$. For each grid cell $m\in\mathcal{B}(\tilde{\mathbf{p}})$,
denote $z_{m}(\tilde{\mathbf{p}})$ as the height of the line segment
of $\tilde{\mathbf{p}}$ that passes over the $m$th grid cell. Denote
$\mathcal{\tilde{D}}_{0}$ as the set of links $\tilde{\mathbf{p}}$
which are not blocked by any virtual obstacle, and $\mathcal{\tilde{D}}_{0}$
is termed as the \ac{los} region in this paper. Similarly, $\mathcal{\tilde{D}}_{1}$
denotes the \ac{nlos} region.

Mathematically, for an \ac{los} case $\tilde{\mathbf{p}}\in\mathcal{\tilde{D}}_{0}$,
we have $v_{m}<z_{m}(\tilde{\mathbf{p}})$ for all grid cells $m\in\mathcal{B}(\tilde{\mathbf{p}})$,
where recall that $v_{m}$ is the height of the virtual obstacle at
the $m$th grid cell; for a \ac{nlos} case $\tilde{\mathbf{p}}\in\mathcal{\tilde{D}}_{1}$,
we have $v_{m}\geq z_{m}(\tilde{\mathbf{p}})$ for at least one grid
cell $m\in\mathcal{B}(\tilde{\mathbf{p}})$, \emph{i.e.}, the direct
path from $\mathbf{p}_{\mathrm{t}}$ to $\mathbf{p}_{\mathrm{r}}$
is blocked by at least one of the virtual obstacles along the direct
path. Thus, a model for the \ac{los} region $\tilde{\mathcal{D}}_{0}(\mathbf{V}_{\textrm{H}})$
given the parameter $\mathbf{V}_{\textrm{H}}$ for the heights of
the virtual obstacles can be formulated as
\begin{equation}
\mathbb{I}\{\tilde{\mathbf{p}}\in\mathcal{\tilde{D}}_{0}(\mathbf{V}_{\textrm{H}})\}=\prod_{m\in\mathcal{B}(\tilde{\mathbf{p}})}\mathbb{I}\{v_{m}<z_{m}(\tilde{\mathbf{p}})\}\label{eq:virtual-obstacle-model}
\end{equation}
and the model for the \ac{nlos} region can also be obtained by $\mathbb{I}\{\tilde{\mathbf{p}}\in\mathcal{\tilde{D}}_{1}(\mathbf{V}_{\textrm{H}})\}=1-\mathbb{I}\{\tilde{\mathbf{p}}\in\mathcal{\tilde{D}}_{0}(\mathbf{V}_{\textrm{H}})\}$,
where $\mathbb{I}\{A\}$ is an indicator function that takes value
$1$ if condition $A$ is true, and $0$, otherwise. As a result,
the blockage-aware channel gain of the direct path of $\tilde{\mathbf{p}}$
is modeled as 
\begin{equation}
g_{\mathrm{0}}(\tilde{\mathbf{p}},\mathbf{V}_{\textrm{H}})=\sum_{k=\{0,1\}}f_{k}(d(\mathbf{p}_{\mathrm{t}},\mathbf{p}_{\mathrm{r}}))\mathbb{I}\{\tilde{\mathbf{p}}\in\tilde{\mathcal{D}}_{k}(\mathbf{V}_{\textrm{H}})\}\label{eq:LOS-component-1}
\end{equation}
where $f_{k}(\cdot)$ are the path gain functions for the \ac{los}
and \ac{nlos} channels, respectively, and $d(\mathbf{p}_{\mathrm{t}},\mathbf{p}_{\mathrm{r}})=\|\mathbf{p}_{\mathrm{t}}-\mathbf{p}_{\mathrm{r}}\|$
represents the distance between the \ac{tx} at $\mathbf{p}_{\mathrm{t}}$
and the \ac{rx} at $\mathbf{p}_{\mathrm{r}}$.

\subsection{Reflection Model}

We exploit the following two properties to construct a lightweight
reflection model with reduced parameters.

First, it practically suffices to consider the reflection from the
main lobe of the beam. Because from the beam pattern model (\ref{eq:beam_pattern})
as illustrated in Fig.~2, there is only one main lobe, and in a typical
massive \ac{mimo} scenario, the side lobes have substantially lower
power than that of the main lobe, which is already very small due
to reflection. Thus, for a link $\tilde{\mathbf{p}}$, the energy
from the reflection of side lobes can be ignored.

Second, in many scenarios, the paths with a single reflection dominate
the energy over paths with multiple reflections.\footnote{Exception scenarios may include tunnels that experience a canyon effect which acts like a waveguide that boosts the signals via multiple reflections. However, such scenarios occur infrequently and are not the focus of this paper.}
This is because paths with multiple reflections tend to travel over
a longer distance and each reflection loses energy.

Therefore, we propose to consider only the paths with a single reflection
and aligned with the main lobe.

\subsubsection{Main Lobe Area}

The main lobe area $\mathcal{\tilde{M}}(\mathbf{p}_{\mathrm{t}},\mathbf{w}_{j})$
is defined as the set of grid cells within the main lobe coverage
of the $j$th beam from the \ac{tx} at $\mathbf{p}_{\text{t}}$.
From the beam pattern model (\ref{eq:beam_pattern}), the main lobe
for the $j$th beamforming vector $\mathbf{w}_{j}$ is defined according
to the beam angle $\theta_{j}$ and the beamwidth $\Theta_{j}$ as
Fig.~\ref{fig:Beam_Pattern}(a).\footnote{In our implementation, we use 3 dB beamwidth $\Theta_j$, which refers to the angular range where the radiated power is within 3 dB of the maximum power.}
Thus, the main lobe area is formulated as
\begin{equation}
\mathbb{I}\{m\in\tilde{\mathcal{M}}(\mathbf{p}_{\mathrm{t}},\mathbf{w}_{j}))=\mathbb{I}\{|\varphi(\mathbf{p}_{\mathrm{t}},\mathbf{c}_{m})-\theta_{j}|\leq\Theta_{j}/2\}.
\end{equation}

\subsubsection{Reflective Zone}

The reflective zone $\tilde{\mathcal{Q}}_{m}(\mathbf{p}_{\mathrm{t}},\mathbf{V})$
of the virtual obstacle at the $m$th grid cell for the \ac{tx} at
$\mathbf{p}_{t}$ is defined as the region of \ac{rx} locations $\mathbf{p}_{\mathrm{r}}$
such that the signal emitted from the \ac{tx} location $\mathbf{p}_{t}$
can be received at \ac{rx} location $\mathbf{p}_{\mathrm{r}}$ via
a single specular reflection at the $m$th virtual obstacle as shown
in Fig.~\ref{reflection1}. It is clear that the reflective zone
depends on the orientation of the virtual obstacle captured by its
normal vector $\mathbf{l}_{m}$, which is a model parameter to be
learned.

Combining the notion of main lobe area and reflective zone, the reflection
via the $m$th virtual obstacle is modeled as
\begin{equation}
\begin{split}g_{\mathrm{r}}^{m}(\tilde{\mathbf{p}},\mathbf{w}_{j},\mathbf{V}) & =\mathbb{I}\{m\in\tilde{\mathcal{M}}(\mathbf{p}_{\mathrm{t}},\mathbf{w}_{j})\}\mathbb{I}\{\mathbf{p}_{\mathrm{r}}\in\tilde{\mathcal{Q}}_{m}(\mathbf{p}_{\mathrm{t}},\mathbf{V})\}\\
 & \quad\times f_{\textrm{r}}(d'(\tilde{\mathbf{p}},\mathbf{c}_{m}),\mathbf{V})
\end{split}
\label{eq:reflection model}
\end{equation}
where $d'\left(\tilde{\mathbf{p}},\mathbf{c}_{m}\right)=d(\mathbf{p}_{\mathrm{t}},\mathbf{c}_{m})+d(\mathbf{p}_{\mathrm{r}},\mathbf{c}_{m})$
represents the propagation distance of the reflected path. The first
condition $\mathbb{I}\{m\in\tilde{\mathcal{M}}(\mathbf{p}_{\mathrm{t}},\mathbf{w}_{j})\}$
restricts the reflected location to the main lobe, the second condition
$\mathbb{I}\{\mathbf{p}_{\mathrm{r}}\in\tilde{\mathcal{Q}}_{m}(\mathbf{p}_{\mathrm{t}},\mathbf{V})\}$
models the reflective zone, and the final term $f_{\textrm{r}}(\cdot)$
captures the path gain when both conditions are satisfied.

The remaining challenge is to design a model for the reflective zone
$\tilde{\mathcal{Q}}_{m}(\mathbf{p}_{\mathrm{t}},\mathbf{V})$. A
more explicit mechanism shall be constructed for $\tilde{\mathcal{Q}}_{m}(\mathbf{p}_{\mathrm{t}},\mathbf{V})$
such that it can be learned efficiently under a deep learning framework.

\subsection{Scattering Model}

While the energy from the direct path and prominent reflections is
captured by the first two terms in (\ref{eq:radio-map-model}), a
small residual may remain due to multiple reflections, diffraction,
and scattering. Explicitly modeling these fine-grained propagation
phenomena is highly challenging. To this end, a generic physical principle
is exploited, where the propagation mainly depends on the local environment
surrounding the \ac{tx} and \ac{rx}, since similar local geometry
can exhibit analogous propagation paths (in terms of direction and
quantity). As such, we adopt a model to map the local environment
to the residual scattering, and more specifically, we employ such
a scattering model from \cite{WangqianChen:J24} and extend it to
the MIMO beam scenario.

Define $\mathcal{B}_{\mathrm{s}}(\tilde{\mathbf{p}},\mathbf{V})$
as the set of grid cells within a local area around $\mathbf{p}_{\mathrm{t}}$
and $\mathbf{p}_{\mathrm{r}}$, where the area is defined as an ellipse
with these positions as foci under a given eccentricity ratio as illustrated
in \cite{WangqianChen:J24}. The scattering model is thus given as
\begin{equation}
\begin{split}g_{\mathrm{s}}(\tilde{\mathbf{p}},\mathbf{w}_{j},\mathbf{V})=f_{s}(\mathcal{B}_{\mathrm{s}}(\tilde{\mathbf{p}},\mathbf{V}),\mathbf{w}_{j})\end{split}
\label{eq:Scattering-component}
\end{equation}
where $f_{s}(\cdot)$ is a mapping that depends on the beam $\mathbf{w}_{j}$
and the local geometric features of $\mathbf{V}$. Furthermore, $f_{s}(\cdot)$
should be invariant to an absolute offset and rotation of the coordinate
system, such as the rotation and translation of $\mathcal{B}_{\mathrm{s}}(\tilde{\mathbf{p}},\mathbf{V})$.

Note that the design of the scattering model in \cite{WangqianChen:J24}
cannot capture the characteristics of MIMO beams, and therefore, the
remaining challenge is to incorporate beam directionality and the
associated beam selectivity into the model.

\section{Geometric Features of Reflective Zone}

In this section, we derive the expression for the reflective zone
condition $\mathbb{I}\{\mathbf{p}_{\mathrm{r}}\in\tilde{\mathcal{Q}}_{m}(\mathbf{p}_{\mathrm{t}},\mathbf{V})\}$.
Subsequently, we analyze its geometric features and develop a simplified
reformulation more compatible with deep learning representations.

\subsection{Analytical Characterization for the Reflective Zone}

We first derive an analytical expression for $\mathbb{I}\{\mathbf{p}_{\mathrm{r}}\in\tilde{\mathcal{Q}}_{m}(\mathbf{p}_{\mathrm{t}},\mathbf{V})\}$
in (\ref{eq:reflection model}) that determines whether a position
$\mathbf{p}_{\mathrm{r}}$ can receive a reflected signal emitted
from \ac{tx} position $\mathbf{p}_{\mathrm{t}}$ via the $m$th virtual
obstacle with a certain orientation.

Using the concept of virtual \ac{tx} in specular reflection, consider
to abstract the virtual obstacle as a rectangular board perpendicular
to the ground with a certain height as shown in Fig.~\ref{reflection1}
in a horizontal view. The orientation of the board is specified by
the normal vector $\mathbf{l}_{m}$. The reflection received at the
\ac{rx} is equivalent to a direct path from the virtual TX$^{m}$,
the mirror of the \ac{tx} with respect to the board, to the \ac{rx}.
As a result, the reflective zone can be understood as the shadow of
the rectangular board being illuminated by the virtual TX$^{m}$.

Mathematically, denote $\mathbf{p}_{\mathrm{t}}^{m}$ as the position
of the virtual \ac{tx} relative to the $m$th grid cell. The board
can be modeled using a 3D plane function as $\mathbf{l}_{m}\cdot(\mathbf{o}_{p}-\mathbf{c}_{m})=0$,
where $\mathbf{o}_{p}$ denotes the point on the board. Similarly,
the virtual direct path is modeled using a 3D line function as $\mathbf{o}_{l}=\mathbf{p}_{\mathrm{r}}+l(\mathbf{p}_{\textrm{t}}^{m}-\mathbf{p}_{\mathrm{r}})$,
where $\mathbf{o}_{l}$ is the point of the line and $l$ denotes
its scale factor. The intersection $\mathbf{o}_{m}$ of the two functions
is identified by 
\begin{equation}
\mathbf{o}_{m}=\mathbf{p}_{\mathrm{r}}+\frac{\mathbf{l}_{m}\cdot(\mathbf{c}_{m}-\mathbf{p}_{\mathrm{r}})}{\mathbf{l}_{m}\cdot(\mathbf{p}_{\textrm{t}}^{m}-\mathbf{p}_{\mathrm{r}})}(\mathbf{p}_{\textrm{t}}^{m}-\mathbf{p}_{\mathrm{r}}).\label{eq:int_point}
\end{equation}
To ensure that the intersection $\mathbf{o}_{m}$ falls within the
$m$th grid cell, we have the orientation conditions $\left|o_{m,i}-c_{m,i}\right|_{i=\{1,2\}}\leq D/2$
under $0\leq(\mathbf{l}_{m}\cdot(\mathbf{c}_{m}-\mathbf{p}_{\mathrm{r}}))/(\mathbf{l}_{m}\cdot(\mathbf{p}_{\textrm{t}}^{m}-\mathbf{p}_{\mathrm{r}}))\leq1$,
and the height condition $o_{m,3}\leq v_{m}$, where $o_{m,i}$ and
$c_{m,i}$ denote the $i$th element of $\mathbf{o}_{m}$ and $\mathbf{c}_{m}$,
respectively. Therefore, the reflective zone condition in (\ref{eq:reflection model})
can be formulated as
\begin{equation}
\begin{split}\mathbb{I}\{\mathbf{p}_{\mathrm{r}}\in\tilde{\mathcal{Q}}_{m}(\mathbf{p}_{\mathrm{t}},\mathbf{V})\}= & \prod_{i=\{1,2\}}\mathbb{I}\{|o{}_{m,i}-c_{m,i}|\leq D/2\}\\
 & \times\mathbb{I}\{0\leq\frac{\mathbf{l}_{m}\cdot(\mathbf{c}_{m}-\mathbf{p}_{\mathrm{r}})}{\mathbf{l}_{m}\cdot(\mathbf{p}_{\textrm{t}}^{m}-\mathbf{p}_{\mathrm{r}})}\leq1\}\\
 & \times\mathbb{I}\{o{}_{m,3}\leq v_{m}\}
\end{split}
\label{eq:reflective_zone_model}
\end{equation}

Although an explicit expression is found in (\ref{eq:reflective_zone_model}),
solving an inverse problem to learn the geometric parameters $v_{m}$
and $\mathbf{l}_{m}$ is still challenging. This is due to the nonlinear
expression for $\mathbf{l}_{m}$ in (\ref{eq:int_point}) and the
multiple indicator functions in (\ref{eq:reflective_zone_model}),
which easily lead to gradient vanishing or explosion in a deep neural
network. To this end, we further simplify (\ref{eq:reflective_zone_model})
by decomposing the terms for the height $v_{m}$ and the orientation
$\mathbf{l}_{m}$ as follows.

\begin{figure}
\centering \subfigure[]{\label{reflection1}\includegraphics[scale=0.45]{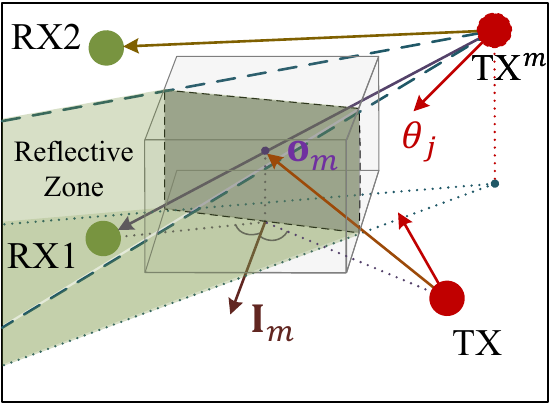}}
\subfigure[]{\label{reflection2} \includegraphics[scale=0.45]{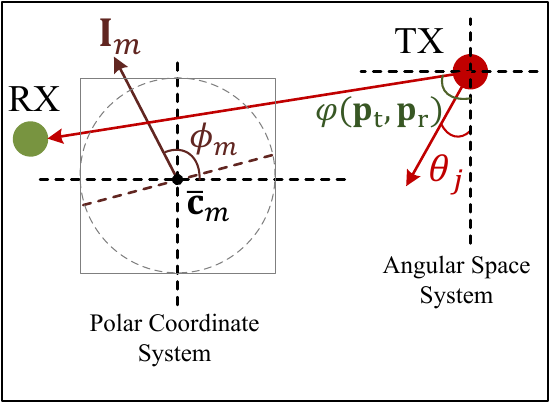}}

\subfigure[]{\label{reflection3}\includegraphics[scale=0.45]{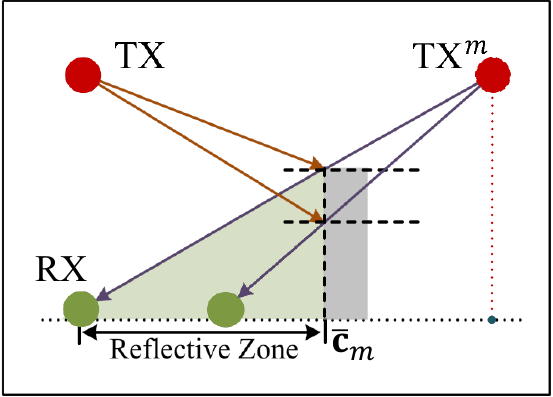}}
\subfigure[]{ \label{reflection4}\includegraphics[scale=0.45]{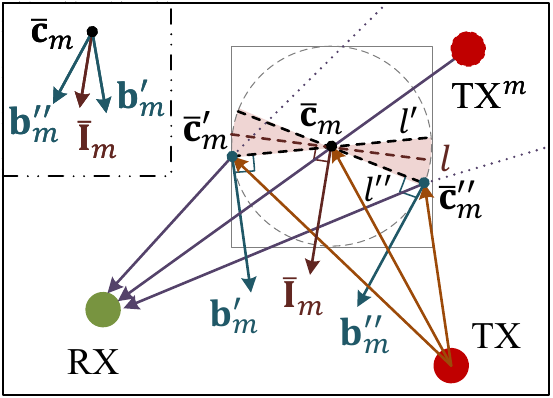}}\caption{a) The diagram of the reflective zone geometry; b) The relationship
between the polar coordinate system and the angular space; c) The
geometry of the height condition for the reflective zone; d) The geometry
of the orientation condition for the reflective zone.}
\label{fig:Reflection_model}
\end{figure}

\subsection{Height Condition for the Reflective Zone}

We show that the term $\mathbb{I}\{o{}_{m,3}\leq v_{m}\}$ in (\ref{eq:reflective_zone_model})
which couples the parameters $v_{m}$ and $\mathbf{l}_{m}$ via $\mathbf{o}_{m}$
in (\ref{eq:int_point}) can be simplified into an expression involving
only the obstacle heigh $v_{m}$.

Denote $\mathbf{c}_{m}=(\mathbf{\bar{c}}_{m},c_{m,3})$, where $\mathbf{\bar{c}}_{m}=(c_{m,1},c_{m,2})\in\mathbb{R}^{2}$
is the ground projected location. Similar notations $\mathbf{p}_{\mathrm{t}}=(\bar{\mathbf{p}}_{\textrm{t}},p_{\textrm{t,3}})$
and $\mathbf{p}_{\mathrm{r}}=(\bar{\mathbf{p}}_{\textrm{r}},p_{\textrm{r,3}})$
are defined for the TX and RX, respectively. By the law of reflection,
the locations $\mathbf{p}_{\textrm{t}}^{m}$, $\mathbf{c}_{m}$ and
$\mathbf{p}_{\mathrm{r}}$ are collinear as shown in Fig.~\ref{reflection3},
where the virtual obstacle in vertical direction is simplified as
a line segment that perpendicular to the ground. Denote $l_{m}=(\mathbf{l}_{m}\cdot(\mathbf{c}_{m}-\mathbf{p}_{\mathrm{r}}))/(\mathbf{l}_{m}\cdot(\mathbf{p}_{\textrm{t}}^{m}-\mathbf{p}_{\mathrm{r}}))$,
we can reformulate it as
\begin{equation}
\begin{aligned}l_{m} & =\left\Vert \mathbf{c}_{m}-\mathbf{p}_{\mathrm{r}}\right\Vert /\left\Vert \mathbf{p}_{\textrm{t}}^{m}-\mathbf{p}_{\mathrm{r}}\right\Vert \\
 & =d\left(\mathbf{\bar{c}}_{m},\mathbf{\bar{p}}_{\mathrm{r}}\right)/(d\left(\mathbf{\bar{c}}_{m},\mathbf{\bar{p}}_{\mathrm{r}}\right)+d\left(\mathbf{\bar{c}}_{m},\mathbf{\bar{p}}_{\mathrm{t}}\right))
\end{aligned}
\end{equation}
Taking this factor $l_{m}$ into (\ref{eq:int_point}), the term $\mathbb{I}\{o{}_{m,3}\leq v_{m}\}$
can be reformulated as $\mathbb{I}\{\widehat{v}_{m}(\tilde{\mathbf{p}})\leq v_{m}\}$,
where
\begin{equation}
\begin{split}\widehat{v}_{m}(\tilde{\mathbf{p}})=\frac{(p_{t,3}-p_{r,3})d\left(\mathbf{\bar{c}}_{m},\mathbf{\bar{p}}_{\mathrm{r}}\right)}{d\left(\mathbf{\bar{c}}_{m},\mathbf{\bar{p}}_{\mathrm{r}}\right)+d\left(\mathbf{\bar{c}}_{m},\mathbf{\bar{p}}_{\mathrm{t}}\right)}+p_{r,3}.\end{split}
\label{eq:Height Condition}
\end{equation}

Noted that the term $\mathbb{I}\{\widehat{v}_{m}(\tilde{\mathbf{p}})\leq v_{m}\}$
is decoupled from $\mathbf{l}_{m}$, as all parameters in (\ref{eq:Height Condition})
are given by the locations of the TX, RX and grid cell, which facilitates
preprocessing of the height condition for $v_{m}$ to accelerate both
training and inference.

\subsection{Orientation Condition for the Reflective Zone}

We further determine the geometric boundaries for $\mathbf{l}_{m}$
to ensure that the intersection $\bar{\mathbf{o}}_{m}=(o_{m,1},o_{m,2})$
satisfies the orientation conditions in (\ref{eq:reflective_zone_model}).

Denote $\mathbf{l}_{m}=(\bar{\mathbf{l}}_{m},\bar{I}_{m,3})$, where
$\bar{\mathbf{l}}_{m}=(\bar{I}_{m,1},\bar{I}_{m,2})\in\mathbb{R}^{2}$
is the ground projected vector. Recall that the virtual obstacle is
abstracted into a rectangular board, which is represented as a line
segment passing through $\mathbf{\bar{c}}{}_{m}$ in a top view as
shown in Fig.~\ref{reflection4}, and it is perpendicular to $\bar{\mathbf{l}}_{m}$.
Denote $\mathbf{\bar{c}}{}_{m}'$ and $\mathbf{\bar{c}}{}_{m}''$
as the two endpoints of the rectangular board shown in Fig.~\ref{reflection4},
and therefore, the reflection point $\bar{\mathbf{o}}_{m}$ must locate
in the line segment between the two endpoints $\mathbf{\bar{c}}{}_{m}'$
and $\mathbf{\bar{c}}{}_{m}''$. Define $\mathbf{b}_{m}'(\tilde{\mathbf{p}})$
as the normal vector of the virtual board when the reflection point
locates at $\mathbf{\bar{c}}{}_{m}'$. It follows that
\begin{equation}
(\mathbf{\bar{c}}{}_{m}'-\mathbf{\bar{c}}{}_{m})\cdot\mathbf{b}_{m}'(\tilde{\mathbf{p}})=0.\label{eq:boundary vectors-1}
\end{equation}
Then, according to the law of specular reflection, the incident angle
$\angle(\mathbf{b}_{m}'(\tilde{\mathbf{p}}),\overrightarrow{\mathbf{\bar{c}}{}_{m}'\mathbf{\bar{p}}_{\textrm{r}}})$
equals to the emergence angle $\angle(\mathbf{b}_{m}'(\tilde{\mathbf{p}}),\overrightarrow{\mathbf{\bar{c}}{}_{m}'\mathbf{\bar{p}}_{\textrm{t}}})$,
leading to the relation
\begin{equation}
\small\begin{split}\mathbf{b}_{m}'(\tilde{\mathbf{p}}) & =(\mathbf{\bar{p}}_{\textrm{r}}-\mathbf{\bar{c}}{}_{m}')/\left\Vert \mathbf{\bar{p}}_{\textrm{r}}-\mathbf{\bar{c}}{}_{m}'\right\Vert +(\mathbf{\bar{p}}_{\textrm{t}}-\mathbf{\bar{c}}{}_{m}')/\left\Vert \mathbf{\bar{p}}_{\textrm{t}}-\mathbf{\bar{c}}{}_{m}'\right\Vert \end{split}
.\label{eq:boundary vectors}
\end{equation}

Similarly, define $\mathbf{b}_{m}''(\tilde{\mathbf{p}})$ as the normal
vector of the virtual board when the point $\bar{\mathbf{o}}_{m}$
locates at $\mathbf{\bar{c}}{}_{m}''$. Intuitively, when the point
$\bar{\mathbf{o}}_{m}$ locates on the line segment between $\mathbf{\bar{c}}{}_{m}'$
and $\mathbf{\bar{c}}{}_{m}''$, the normal vector $\bar{\mathbf{l}}_{m}$
lies in the convex hull of $\mathbf{b}_{m}'(\tilde{\mathbf{p}})$
and $\mathbf{b}_{m}''(\tilde{\mathbf{p}})$. This result is formally
proven in the following lemma.
\begin{lem}
\label{thm:Lemma}If there is a path $\tilde{\mathbf{p}}$ reflected
by the $m$th obstacle, the vectors $\mathbf{b}_{m}'(\tilde{\mathbf{p}})$,
$\mathbf{b}_{m}''(\tilde{\mathbf{p}})$ and the normal vector $\bar{\mathbf{l}}_{m}$
satisfy
\begin{equation}
\angle(\mathbf{b}_{m}'(\tilde{\mathbf{p}}),\bar{\mathbf{l}}_{m})+\angle(\mathbf{b}_{m}''(\tilde{\mathbf{p}}),\bar{\mathbf{l}}_{m})=\angle(\mathbf{b}_{m}'(\tilde{\mathbf{p}}),\mathbf{b}_{m}''(\tilde{\mathbf{p}}))\label{eq:horizontal_Condition}
\end{equation}
where $\angle(\boldsymbol{x},\boldsymbol{y})=\arccos\left((\boldsymbol{x}\cdot\boldsymbol{y}\mathbf{)}/(\left\Vert \boldsymbol{x}\right\Vert \left\Vert \boldsymbol{y}\right\Vert )\right)$
represents the angle between the two vectors $\boldsymbol{x}$ and
$\boldsymbol{y}$.
\end{lem}
\begin{proof}
See Appendix~\ref{app:Lemma 1 proof}.
\end{proof}
By Lemma~\ref{thm:Lemma}, the orientation conditions in (\ref{eq:reflective_zone_model})
can be directly verified using the conclusion in (\ref{eq:horizontal_Condition}),
thereby bypassing the need for the calculation in (\ref{eq:int_point}).
Meanwhile, as $\mathbf{b}_{m}'(\tilde{\mathbf{p}})$ and $\mathbf{b}_{m}''(\tilde{\mathbf{p}})$
are independent of the trainable variables $v_{m}$ and $\bar{\mathbf{l}}_{m}$,
they are constants that can be precomputed before training.

Using the auxiliary vectors $\widehat{v}_{m}(\tilde{\mathbf{p}})$,
$\mathbf{b}_{m}'(\tilde{\mathbf{p}})$ and $\mathbf{b}_{m}''(\tilde{\mathbf{p}})$,
we now can derive a simplified expression for characterizing the reflective
zone (\ref{eq:reflective_zone_model}) as follows.
\begin{prop}
\label{thm:Proposed1}Equation (\ref{eq:reflective_zone_model}) can
be simplified to
\begin{equation}
\begin{split}\mathbb{I}\{\mathbf{p}_{\mathrm{r}}\in\tilde{\mathcal{Q}}_{m}(\mathbf{p}_{\mathrm{t}},\mathbf{V})\} & =\mathbb{I}\{\bar{\mathbf{l}}_{m}\cdot\mathbf{\widehat{b}}_{m}(\tilde{\mathbf{p}})\geq\mathbf{b}_{m}'(\tilde{\mathbf{p}})\cdot\mathbf{\widehat{b}}_{m}(\tilde{\mathbf{p}})\}\\
 & \quad\times\mathbb{I}\{v_{m}\geq\widehat{v}_{m}(\tilde{\mathbf{p}})\}
\end{split}
\label{eq:reflective_zone_model-r1}
\end{equation}
where $\mathbf{\widehat{b}}_{m}(\tilde{\mathbf{p}})=(\mathbf{b}_{m}'(\tilde{\mathbf{p}})/\left\Vert \mathbf{b}_{m}'(\tilde{\mathbf{p}})\right\Vert +\mathbf{b}_{m}''(\tilde{\mathbf{p}})/\left\Vert \mathbf{b}_{m}''(\tilde{\mathbf{p}})\right\Vert )/2$
denotes the bisector vector of $\mathbf{b}_{m}'(\tilde{\mathbf{p}})$
and $\mathbf{b}_{m}''(\tilde{\mathbf{p}})$.
\end{prop}
\begin{proof}
See Appendix~\ref{app:Proposition 1 proof}.
\end{proof}
It is observed that Proposition~\ref{thm:Proposed1} has substantially
simplified the expression of the reflective zone. First, the computation
of $\mathbf{o}_{m}$ in (\ref{eq:int_point}) is not required and
the number of indicator functions has been reduced, which helps reduce
the risk of gradient vanishing or explosion. Second, the parameters
$v_{m}$ and $\mathbf{l}_{m}$ are separated in the two terms in (\ref{eq:reflective_zone_model-r1}),
and thus, the coupling has been reduced. Third, the condition parameters
$\mathbf{b}_{m}'(\tilde{\mathbf{p}})$, $\mathbf{\widehat{b}}_{m}(\tilde{\mathbf{p}})$
and $\widehat{v}_{m}(\tilde{\mathbf{p}})$ in (\ref{eq:reflective_zone_model-r1})
are all known for the link $\tilde{\mathbf{p}}$ and can be precomputed
to accelerate training.

\begin{figure*}[!t]
\centering \includegraphics[scale=0.62]{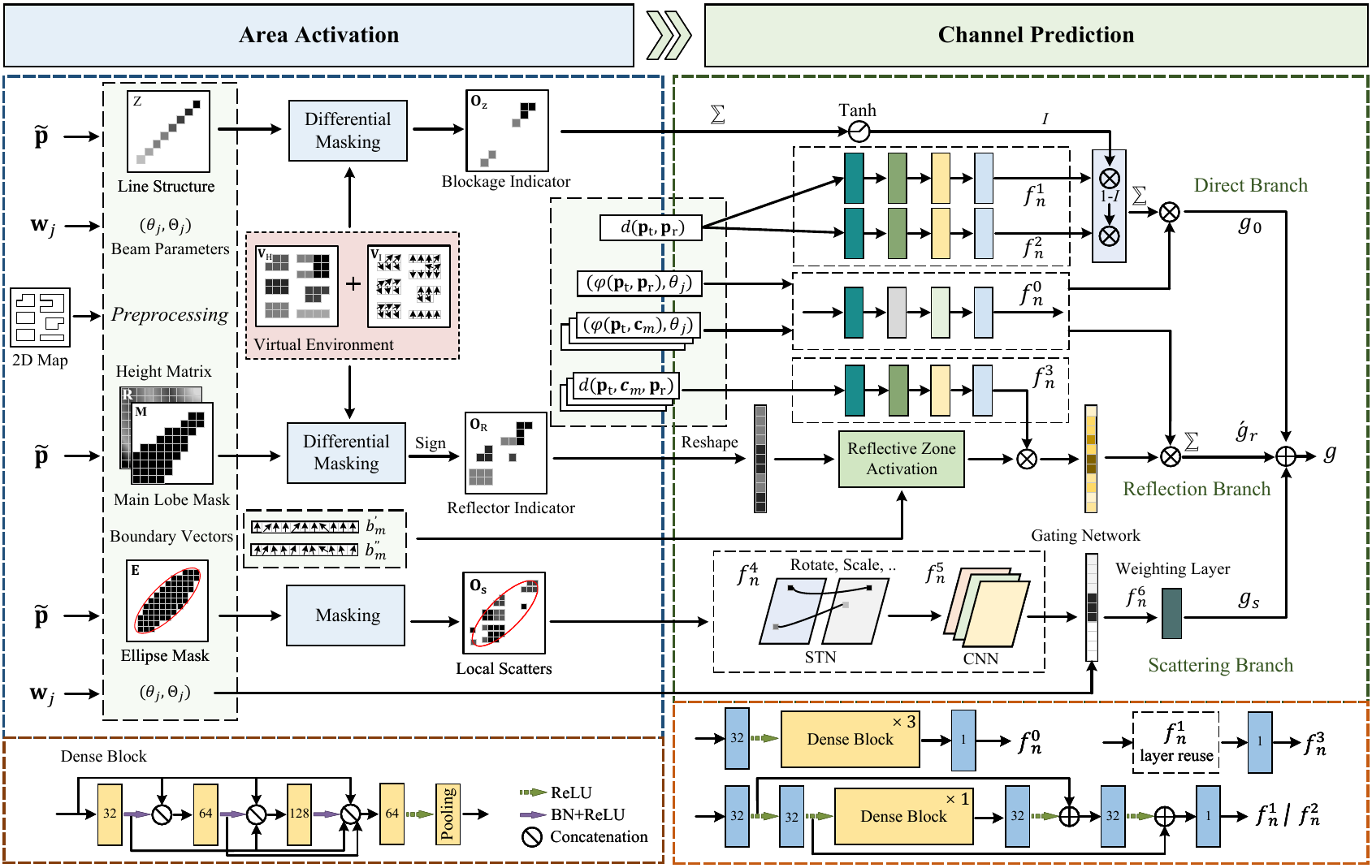} \caption{The physics-informed neural network architecture for joint MIMO beam
map and virtual environment reconstruction.}
\label{fig:network-architecture}
\end{figure*}

\section{Network Architecture and Training}

In this section, we develop a neural network architecture tailored
to the proposed propagation components for joint environmental geometry
and MIMO beam map construction. Specifically, we transform the blockage
relation (\ref{eq:virtual-obstacle-model}) and reflective zone formulation
(\ref{eq:reflective_zone_model-r1}) into neural network representations,
enabling efficient feature propagation and gradient flow for training.
Specialized network structures incorporating dense connection, layer
reuse, and modular sub-networks would be explored to learn distinct
propagation characteristics, which explicitly leverage the geometric
priors embedded in the physical modeling process. The proposed physics-informed
neural network architecture is shown in Fig.~\ref{fig:network-architecture}.

\subsection{Overall Architecture}

As illustrated in Fig.~\ref{fig:network-architecture}, the proposed
network consists of two main modules, \emph{Area Activation} and \emph{Channel
Prediction}. The Area Activation module is designed to identify the
relevant local virtual obstacles for the link $\tilde{\mathbf{p}}$
corresponding to each propagation component, while the Channel Prediction
module models their interactions for accurate channel prediction.

\subsubsection{Area Activation}

This module selects a subset of virtual obstacles relevant to the
link $\tilde{\mathbf{p}}$ and the beamforming vector $\mathbf{w}_{j}$.
As a result, only the subset of the selected virtual obstacles are
activated. For this purpose, four distinct area filters are defined,
each structured as an $M_{1}$ $\times$ $M_{2}$ matrix, matching
the size of the virtual environment.

For the blockage component, a 3D line structure $\mathbf{Z}$ is constructed,
where we project the altitude of the direct line connecting $\mathbf{p}_{\text{t}}$
and $\mathbf{p}_{\text{r}}$ to the ground. Thus, the $m$th element
of the matrix $\mathbf{Z}$ is given by $Z_{m}=(p_{\textrm{t,3}}-p_{\textrm{r,3}})\frac{\left\Vert \bar{\mathbf{c}}_{m}-\bar{\mathbf{p}}_{\textrm{r}}\right\Vert }{\left\Vert \bar{\mathbf{p}}_{\textrm{t}}-\bar{\mathbf{p}}_{\textrm{r}}\right\Vert }+p_{\textrm{r,3}}$,
if $m\in\mathcal{B}(\tilde{\mathbf{p}})$; and $Z_{m}=0$, for $m\notin\mathcal{B}(\tilde{\mathbf{p}})$.
The feature map of Area Activation for the blockage component is obtained
by 
\begin{equation}
\begin{split}\mathscr{\mathbf{\mathscr{\mathbf{O}}}}_{\textrm{Z}}=\textrm{ReLU}((\mathbf{V}_{\textrm{H}}-\mathbf{Z})\odot\mathbf{Z})\end{split}
\label{eq:obstruction_feature_map}
\end{equation}
where $\textrm{ReLU}(x)=\max(x,0)$ and $\odot$ is the Hadamard product.
The non-zero elements of $\mathscr{\mathbf{\mathscr{\mathbf{O}}}}_{\textrm{Z}}$
indicate the locations where the propagation is blocked along the
direct path.

For the reflection component, a main lobe mask $\mathbf{M}$ and a
height matrix $\mathbf{R}$ are defined. In the main lobe mask $\mathbf{M}$,
each entry $M_{m}$ takes value 1 if $\left|\varphi(\mathbf{p}_{\text{t}},\mathbf{c}_{m})-\theta_{j}\right|\leq\Theta_{j}/2$;
otherwise it takes a value 0. In the height matrix $\mathbf{R}$,
the $m$th entry denotes the height condition $\widehat{v}_{m}(\tilde{\mathbf{p}})$
from (\ref{eq:Height Condition}) at the $m$th grid cell for the
link $\tilde{\mathbf{p}}$. Thus, the feature map of Area Activation
for the reflection component is given by 
\begin{equation}
\begin{split}\mathscr{\mathbf{\mathscr{\mathbf{O}}}}_{\textrm{R}}=\mathrm{\textrm{sign}}(\textrm{ReLU}(\mathbf{V}_{\textrm{H}}-\mathbf{R}))\odot\mathbf{M}\end{split}
\label{eq:reflection_feature_map}
\end{equation}
where $\textrm{sign}(x)=\mathbb{I}\{x>0\}$. The non-zero elements
of $\mathscr{\mathbf{\mathscr{\mathbf{O}}}}_{\textrm{R}}$ indicate
the potential reflection locations within the main lobe.

For the scattering component, an ellipse mask $\mathbf{E}$ is constructed
to select local obstacles for a link $\tilde{\mathbf{p}}$. Each entry
$E_{m}$ takes value 1 if $m\in\mathcal{B}_{\mathrm{s}}(\tilde{\mathbf{p}},\mathbf{V})$,
and 0 otherwise. The feature map of Area Activation for the scattering
is given by 
\begin{equation}
\begin{split}\mathscr{\mathbf{\mathscr{\mathbf{O}}}}_{\textrm{S}}=\mathbf{V}_{\textrm{H}}\odot\mathbf{E}\end{split}
\label{eq:scattering_feature_map}
\end{equation}
where the non-zero elements of $\mathscr{\mathbf{\mathscr{\mathbf{O}}}}_{\textrm{S}}$
correspond to the local geometric features of the selected obstacles.

\subsubsection{Channel Prediction}

It includes four branches. The first branch takes the feature map
$\mathbf{O}_{\text{Z}}$ to predict the path loss of the direct path.
The second branch learns the beam pattern. The third branch takes
the feature map $\mathbf{O}_{\text{R}}$ to form the reflective zone
and estimate the reflected paths within it. The final branch takes
the feature map $\mathscr{\mathbf{\mathscr{\mathbf{O}}}}_{\textrm{S}}$
to capture the residual fluctuation due to scattering. The details
are illustrated as follows.

\subsection{The Beam Pattern Gain}

The beam pattern $B(\cdot)$ is learned using a fully connected dense
neural network (FC-DsNet). As illustrated in Fig.~\ref{fig:network-architecture},
the architecture comprises a feedforward layer, followed by a stack
of three identical dense blocks, and concludes with a final feedforward
layer. Each dense block contains three feedforward layers, with a
batch normalization (BN) layer and a ReLU layer between them. Through
dense connections, each layer receives and concatenates feature maps
from all preceding layers, and then forwards them to the next layer,
which facilitates direct gradient flow between layers, mitigating
the vanishing gradient problem and enhancing information flow throughout
the network. Finally, a feedforward layer merges the features, followed
by a pooling layer that halves the feature dimensions. For an arbitrary
location $\boldsymbol{x}$ relative to $\mathbf{p}_{\mathrm{t}}$,
the beam pattern component is implemented by
\begin{equation}
\begin{split}B(\varphi(\mathbf{p}_{\mathrm{t}},\boldsymbol{x}),\mathbf{w}_{j})=f_{n}^{0}([\varphi(\mathbf{p}_{\mathrm{t}},\boldsymbol{x});\theta_{j}])\end{split}
\label{eq:beamforming_gain_nn}
\end{equation}
where $f_{n}^{0}(\cdot)$ represents the mapping of the proposed FC-DsNet,
and $[\varphi(\mathbf{p}_{\mathrm{t}},\boldsymbol{x});\theta_{j}]$
denotes the concatenation of the azimuth angle $\varphi(\mathbf{p}_{\mathrm{t}},\boldsymbol{x})$
and the beam angle $\theta_{j}$.

\subsection{The Blockage-Aware Direct Branch}

The blockage branch is designed to implement the blockage relation
(\ref{eq:virtual-obstacle-model}) to characterize LOS and NLOS conditions,
which is approximated by a tanh function to ensure a non-degenerated
gradient for training. The design of the soft indicator function for
the LOS region $\tilde{\mathcal{D}}_{0}(\mathbf{V}_{\textrm{H}})$
is given as
\begin{equation}
\begin{split}I=1-\tanh(\mathrm{sum}(\mathscr{\mathbf{\mathscr{\mathbf{O}}}}_{\textrm{Z}})\sigma_{0})\end{split}
\label{eq:obstruction_level}
\end{equation}
where $\mathrm{sum}(\mathbf{A})=\sum_{ij}a_{ij}$, $\tanh(x)=(e^{x}-e^{-x})/(e^{x}+e^{-x})$,
and $\sigma_{0}$ is a scalar parameter. Consequently, when the feature
map $\mathscr{\mathbf{\mathscr{\mathbf{O}}}}_{\textrm{Z}}$ contains
all zeros, the direct path is not blocked by any virtual obstacle,
we have the indicator $I=1$. With the indicator $I$, the blockage-aware
channel model of the direct path in (\ref{eq:LOS-component-1}) is
implemented using two parallel FC-DsNets as
\begin{equation}
\begin{split}g_{\mathrm{0}}(\tilde{\mathbf{p}},\mathbf{V}_{\textrm{H}})=If_{n}^{1}(d(\mathbf{p}_{\mathrm{t}},\mathbf{p}_{\mathrm{r}}))+(1-I)f_{n}^{2}(d(\mathbf{p}_{\mathrm{t}},\mathbf{p}_{\mathrm{r}}))\end{split}
\label{eq:LOS-component-1-1}
\end{equation}
where $f_{n}^{1}(\cdot)$ and $f_{n}^{2}(\cdot)$ denote the two parallel
FC-DsNets. The networks $f_{n}^{1}(\cdot)$ and $f_{n}^{2}(\cdot)$
share the same architecture but differ from $f_{n}^{0}(\cdot)$ as
they serve distinct tasks. As illustrated in Fig.~\ref{fig:network-architecture},
the network architecture includes two feedforward layers, followed
by a dense block, and ends with three feedforward layers. The feedforward
layers at both ends of the dense block are connected via skip connections.
This structure can help reduce network complexity compared to fully
dense connections while still preserving effective gradient propagation.

\subsection{The Reflection Branch}

The reflection branch is designed to aggregate reflected paths within
the reflective zones. Define $\mathbf{B}',\mathbf{\widehat{B}}\in\mathbb{R}^{M_{1}M_{2}\times2}$
as the matrix representations of all $\mathbf{b}_{m}'(\tilde{\mathbf{p}})$
and $\mathbf{\widehat{b}}_{m}(\tilde{\mathbf{p}})$ for the link $\tilde{\mathbf{p}}$,
respectively. For operation alignment, a similar notation $\mathbf{\mathbf{V}_{\textrm{I}}}\in\mathbb{R}^{M_{1}M_{2}\times2}$
is also constructed for the normal vector $\bar{\mathbf{l}}_{m}$
for all grid cells. The neural network representation of the reflective
zone in (\ref{eq:reflective_zone_model-r1}) is thus implemented by
\begin{equation}
\mathbf{r}=\tanh(\mathrm{ReLU}(F(\mathbf{B}',\mathbf{\widehat{B}},\mathbf{\mathbf{V}_{\textrm{I}}}))\overline{\mathscr{\mathbf{\mathscr{\mathbf{O}}}}_{\textrm{R}}}\sigma_{1})\label{eq:Ref_Zone_NN}
\end{equation}
and
\begin{equation}
F(\mathbf{B}',\mathbf{\widehat{B}},\mathbf{\mathbf{V}_{\textrm{I}}})=\mathrm{sumc}(\mathbf{\mathbf{V}_{\textrm{I}}}\odot\mathbf{\widehat{B}})-\mathrm{sumc}(\mathbf{B}'\odot\mathbf{\widehat{B}})\label{eq:Ref_Zone_NN-1}
\end{equation}
where $\overline{\mathscr{\mathbf{\mathscr{\mathbf{O}}}}_{\textrm{R}}}=\mathrm{vec}(\mathscr{\mathbf{\mathscr{\mathbf{O}}}}_{\textrm{R}})$
denotes the column vectorization of a matrix, $\sigma_{1}$ is a scalar
parameter, and $\mathrm{sumc(\cdot)}$ represents the row-wise summation
operation. For the impact of blockage, define $\varGamma_{m}$ as
the LOS indicator that takes value 1 if neither the link $(\mathbf{p}_{\mathrm{t}},\mathbf{c}_{m})$
nor $(\mathbf{c}_{m},\mathbf{p}_{\mathrm{r}})$ is blocked by any
virtual obstacle, as the implementation in (\ref{eq:obstruction_level}).
As a result, the reflection via the $m$th oriented virtual obstacle
in (\ref{eq:reflection model}) is obtained by
\begin{equation}
\begin{split}g_{\mathrm{r}}^{m}(\tilde{\mathbf{p}},\mathbf{w}_{j},\mathbf{V})=r_{m}\varGamma_{m}f_{n}^{3}(d'(\tilde{\mathbf{p}},\mathbf{c}_{m}))+(1-r_{m}\varGamma_{m})\tau\end{split}
\label{eq:reflection_nn}
\end{equation}
where $r_{m}$ denotes the $m$th element of vector $\mathbf{r}$,
$f_{n}^{3}(\cdot)$ is the neural network, and $\tau$ specifies the
minimum cutoff threshold for gain. The network $f_{n}^{3}(\cdot)$
reuses all network layers of $f_{n}^{1}(\cdot)$ and adds a feedforward
layer to predict the final path gain. This design implies that the
signal attenuation of reflected paths and direct paths in free space
share some similar characteristics.

\subsection{The Scattering Branch}

The scattering branch is designed to capture the spatial invariance
property of the model (\ref{eq:Scattering-component}) and to map
the local environment to residual scattering. We adopt the neural
network from \cite{WangqianChen:J24}, which integrates a spatial
transformation network (STN) to capture rotation and scaling invariance,
and a CNN to extract structural similarities from the local geometry
of $\mathbf{V}$ for scattering mapping. Note that the design in \cite{WangqianChen:J24}
does not account for beam characteristics or support beam selection.
To address this limitation, we extend its output to incorporate beam
expansion and design a gating network to achieve beam selectivity.
Thus, the scattering implementation is given by
\begin{equation}
\begin{split}g_{\mathrm{s}}(\tilde{\mathbf{p}},\mathbf{w}_{j},\mathbf{V})=f_{n}^{6}(f_{n}^{5}(f_{n}^{4}(\mathscr{\mathbf{\mathscr{\mathbf{O}}}}_{\textrm{S}}))H(\mathbf{w}_{j}))\end{split}
\label{eq:reflection_nn-1}
\end{equation}
where $f_{n}^{4}(\cdot)$, $f_{n}^{5}(\cdot)$ denote the STN and
CNN, respectively, $H(\mathbf{w}_{j})$ is a gating network that selects
the contributing beams, and $f_{n}^{6}(\cdot)$ is a weighting layer
for the selected beams.

Specifically, we extend the output dimension of the CNN to $|\mathcal{W}|$,
each corresponding to a specific beam in the codebook. The gating
network $H(\mathbf{w}_{j})$ then produces a multi-hot output that
selects index $j$ and its neighbors, while setting all other entries
to zero. The final weighting layer $f_{n}^{6}(\cdot)$ learns to assign
weights to the selected beams through training.

\subsection{Neural Network Training}

We implement the proposed network using the deep learning library
Pytorch and train it in a supervised manner using the Adam optimizer.
Moreover, we employ a hybrid learning algorithm that combines particle
swarm optimization (PSO) with Adam algorithm to update the parameter
$\mathbf{\mathbf{V}_{\textrm{I}}}$. In each iteration, PSO introduces
perturbations to the current solution for a broader global exploration,
followed by Adam algorithm for local fine-tuning. This hybrid strategy
combines the global exploration capability of PSO to escape local
optima and the fine-tuning efficiency of Adam algorithm for convergence.

\section{Simulations}

The experiments are conducted using simulated channel data generated
by a \ac{rt} software,\emph{ Remcom Wireless Insite}.

In the simulation, the region of interest is a 640~m $\times$ 640~m
square area with 9 \acpl{tx} mounted at a height of 50~m. Each TX
is equipped with a 16-element \ac{mimo} antenna array, and the DFT
codebook is used. Ground \acpl{rx} are fixed at a height of 2~m.
Several buildings are included to enhance channel diversity. Channel
data are generated with up to 6 reflections and 1 diffraction. A sinusoidal
waveform at 2.8 GHz with 100 MHz bandwidth is used. In total, over
120 beam maps are generated, comprising more than 500,000 channel
samples that capture rich propagation characteristics, including blockage
and reflection, thereby providing a comprehensive dataset for modeling
and analysis.

The proposed model is compared with the following baselines, RadioUNet
\cite{LevRonYap:J21}, RME-GAN \cite{ZhaWij:J23}, DeepCom \cite{TegRom:J21},
SVT \cite{CaiCand:J20}, and KNN, which are summarized as follows:
\begin{enumerate}
\item RadioUNet: This network consists of two UNets. The first takes the
city map, sparse observations, and TX position as inputs to produce
a coarse beam map, which is then refined by the second UNet using
both the coarse map and original inputs to generate the final beam
map.
\item RME-GAN: This network utilizes UNet as the generator, which takes
the city map, sparse observations and TX position as inputs for beam
map generation.
\item DeepCom: This network employs a CNN-based encoder-decoder architecture
that compresses the sparse observation map into a low-dimensional
latent representation and then reconstructs the complete beam map
from it.
\item SVT: The algorithm recovers a matrix by minimizing its nuclear norm,
iteratively applying singular value decomposition and soft-thresholding
to converge towards the original low-rank matrix.
\item KNN: The algorithm selects 6 measurement samples that are closest
to $\tilde{\mathbf{p}}$ and forms the neighbor set $\mathcal{N}(\tilde{\mathbf{p}})$.
The channel gain at $\tilde{\mathbf{p}}$ is given by $g(\tilde{\mathbf{p}})=\mu^{-1}\sum_{i\in\mathcal{N}(\tilde{\mathbf{p}})}\omega(\tilde{\mathbf{p}},\tilde{\mathbf{p}}^{(i)})y^{(i)}$,
where $\omega(\tilde{\mathbf{p}},\tilde{\mathbf{p}}^{(i)})=\exp[-\kappa||\tilde{\mathbf{p}}-\tilde{\mathbf{p}}^{(i)}||_{2}^{2}]$
and $\mu=\sum_{i\in\mathcal{N}(\tilde{\mathbf{p}})}\omega(\tilde{\mathbf{p}},\tilde{\mathbf{p}}^{(i)})$.
\end{enumerate}

Specifically, data from 5 TXs, with nearly 30\% of measurement samples
from each unless otherwise specified, is used for training, while
the remaining 4 TXs are reserved for testing. Since our focus is on
high-energy regions, such as the main lobe of the beam, we truncate
channel gains below a cutoff threshold for performance evaluation.
The truncation function is defined as $\zeta(x)=\max(x,\eta)$, where
$\eta$ is a cutoff threshold, and it is set to $-130$ dB for all
experiments. The mean absolute error (MAE) and the root mean square
error (RMSE) between the truncated simulated gain $\zeta(y^{(i)})$
and the estimated gain $\zeta(g(\tilde{\mathbf{p}}^{(i)},\mathbf{w}_{j};\bm{\Theta},\mathbf{V}))$
will serve as performance metrics. This approach aligns with practical
cases where beamforming is designed to ensure high communication quality
within a target region, while performance outside it is less critical.

\subsection{Performance of the Proposed Reflection Component}

In this experiment, we evaluate the improvement in \ac{mimo} beam
map construction with the reflection component. 

Table~\ref{Tab:example_with_without_ref.} presents a comparison
of the \ac{mimo} beam map accuracy achieved by the proposed model
with different design components involving blockage, reflection and
scattering. The results show that incorporating both direct and reflected
paths significantly improves performance compared to using only the
direct branch, which results in a $20.7\%$ reduction in MAE and a
$22.4\%$ reduction in RMSE. In addition, further extending the model
to account for scattering leads to additional performance gains, with
MAE and RMSE further reduced by $19.1\%$ and $11.6\%$, respectively,
thereby demonstrating the effectiveness of the proposed physics-informed
components.
\begin{figure}[!t]
\centering \subfigure[]{\includegraphics[scale=0.325]{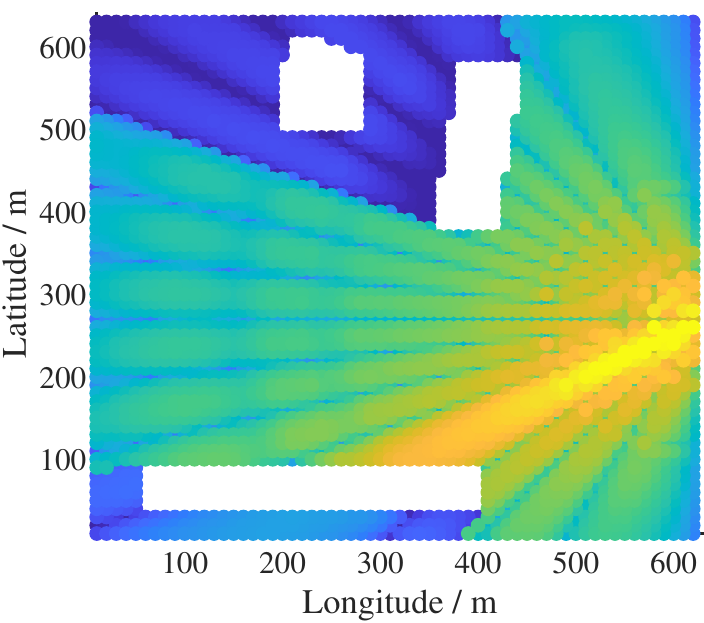}}\hspace{0.5cm}\subfigure[]{\includegraphics[scale=0.32]{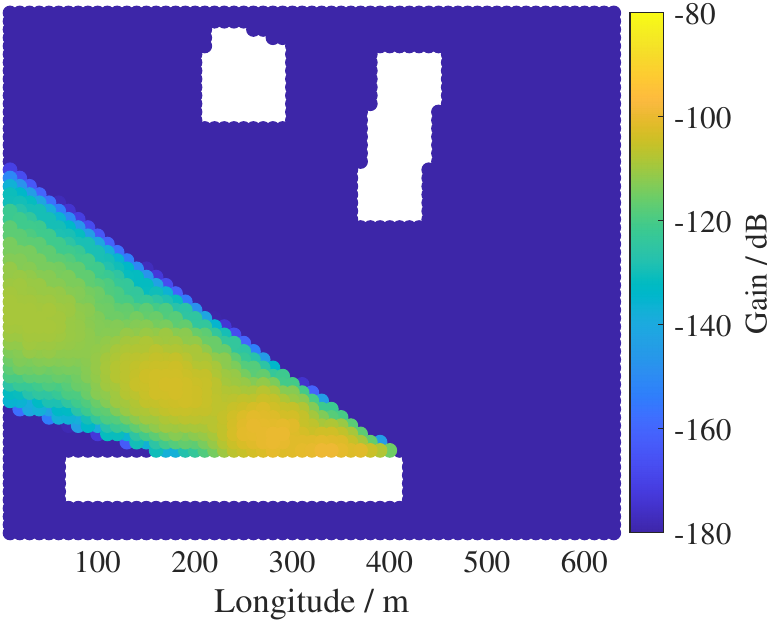}}

\subfigure[]{\label{fig:RadioMap-Example-LOS}\includegraphics[scale=0.32]{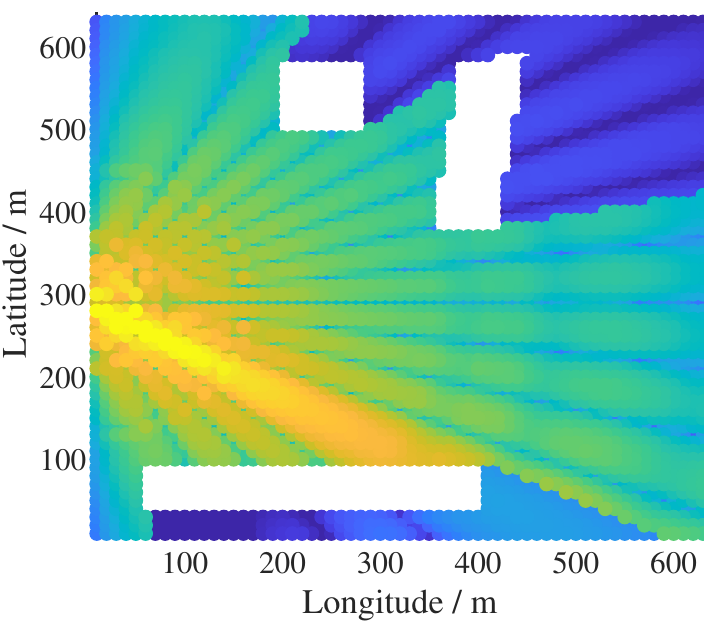}}\hspace{0.5cm}\subfigure[]{\label{fig:RadioMap-Example-Ref}\includegraphics[scale=0.32]{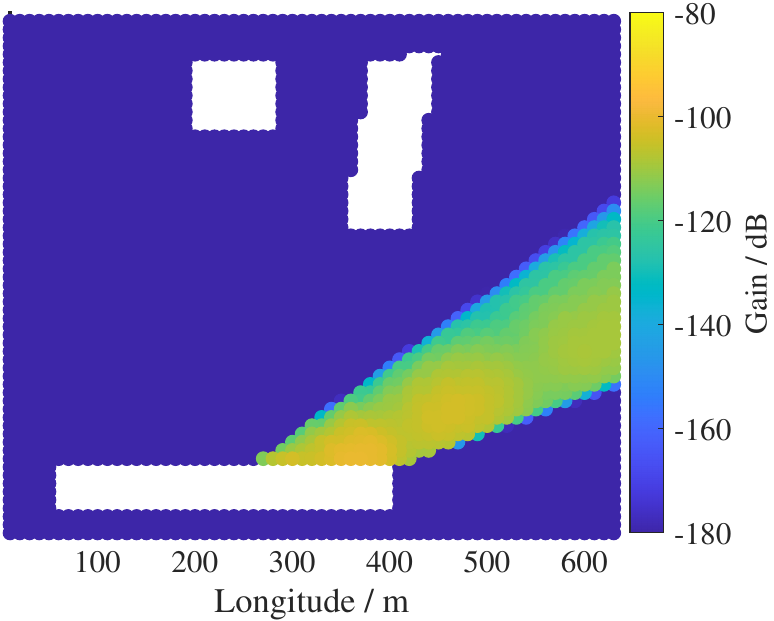}}\caption{The direct and reflected beams are separated: a) and c) The beams
from two TXs are blocked; b) and d) The corresponding reflected beams.}
\label{fig:RadioMap-Example}
\end{figure}

\begin{table}
\caption{Comparison of MIMO beam map with different design branches}

\centering\label{Tab:example_with_without_ref.} \renewcommand\arraystretch{1.5}
\begin{tabular}{ p{3.7cm}<{\centering}| p{1.90cm}<{\centering}| p{1.90cm}<{\centering}}   
\hline   

Scheme          & MAE / dB        & RMSE / dB \\  
  
\hline   
Direct branch    &3.144  &5.737    \\   
\hline

Only direct + reflection  &\textbf{2.491}   &\textbf{4.461}  \\   
\hline

Direct + reflection + scattering  &\textbf{2.016}   &\textbf{3.941}  \\   
\hline

\end{tabular} 
\end{table}

Fig.~\ref{fig:RadioMap-Example} illustrates the direct and reflected
beams generated by the proposed model. The results show that the proposed
model can reconstruct the radio beam geometry, including beam shape,
direction, blockage, and reflection. For a specific obstacle, the
proposed model can capture direct beams impinging on the obstacle
from different \acpl{tx} and directions, along with their corresponding
reflected beams. As a result, the blockage and reflection components
of a beam can be separated from only \ac{rss} measurements. This
capability to identify blockage and reflection provides practical
advantages in real-world applications, as further discussed in this
paper.

\subsection{Geometry of the Virtual Environment}

We further illustrate the geometry of the environment captured by
$\mathbf{V}$. Specifically, we assume the obstacle locations are
known from a 2D map, which lacks height and orientation information.
Each obstacle is further divided into smaller parts, with each part
potentially spanning multiple grid cells, allowing each part to independently
learn its own parameters.

Fig.~\ref{fig:3D terrain map} shows the 3D terrain of the simulated
environment along with the deployment of the \acpl{tx}, while the
corresponding oriented virtual obstacle map, incorporating both the
height information $\mathbf{V}_{\textrm{H}}$ and angle information
$\mathbf{V_{\mathrm{\Phi}}}$ is shown in Fig.~\ref{fig:Height_Map}.
The results indicate that the geometry and distribution of the obstacles
learned by the proposed model closely align with the simulated environment,
demonstrating its ability to capture environmental geometry from \ac{rss}
measurements.

Note that city maps usually do not correspond to the desired virtual
environment maps, as they typically lack critical radio information
such as electromagnetic coefficients. In Fig.~\ref{fig:Height_Map},
the height $\mathbf{V}_{\textrm{H}}$ does not have to exactly match
the actual obstacle heights; instead, it is optimized so that the
channel model (\ref{eq:radio-map-model}) fits the \ac{rss} measurements.
Consequently, different parts of the same obstacle may be assigned
different heights. The normal angles $\mathbf{V_{\mathrm{\Phi}}}$
for each obstacle align with their orientations in the simulated environment,
and they also remain largely consistent across different parts of
the same obstacle.

\begin{figure}[!t]
\centering\subfigure[]{\label{fig:3D terrain map}\includegraphics[scale=0.33]{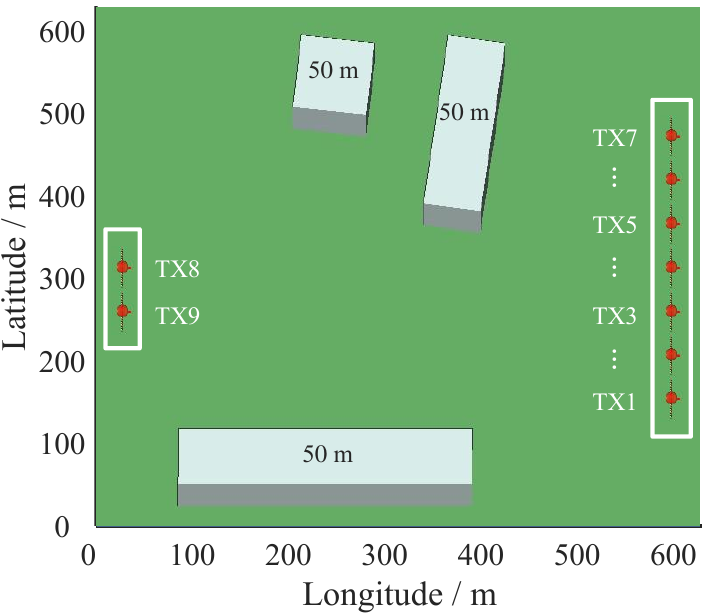}
}\subfigure[]{\label{fig:Height_Map}\includegraphics[scale=0.335]{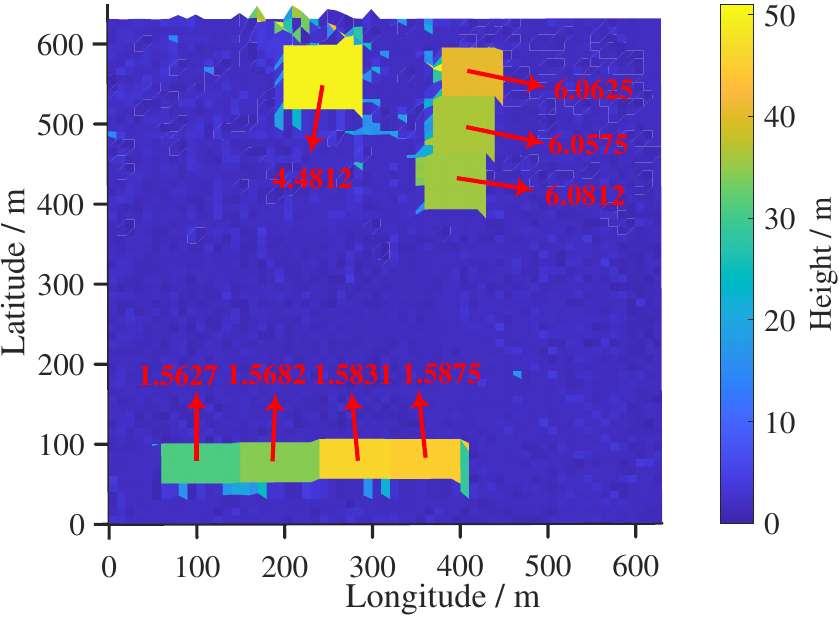}
}\caption{a) The 3D terrain of the simulated environment, where odd-numbered
TXs (e.g., 1, 3, 5) are used for training and even-numbered TXs (e.g.,
2, 4, 6) are used for testing; b) The oriented virtual obstacle map
learned by the proposed model.}
\label{fig:Geometry of the Virtual Environment}
\end{figure}

\begin{table}
\caption{Comparison of beam map accuracy with different methods.}

\centering\label{Tab:example_com_different_Alg.} \renewcommand\arraystretch{1.5}
\begin{tabular}{ p{3cm}<{\centering}| p{2.25cm}<{\centering}| p{2.25cm}<{\centering}}   
\hline   

Scheme          & MAE / dB        & RMSE / dB \\  
  
\hline   
RadioUNet \cite{LevRonYap:J21}            & 3.171          & 5.763     \\   
\hline 

RME-GAN \cite{ZhaWij:J23}                 & 3.889          & 6.499     \\   
\hline   

DeepCom \cite{TegRom:J21}                 & 2.975          & 5.235     \\   
\hline   

SVT \cite{CaiCand:J20}                    & 4.804             & 8.647     \\   
\hline  

KNN                     & 3.360             & 5.892     \\   
\hline

Proposed              &\textbf{2.016}    &\textbf{3.941}    \\   
\hline

\end{tabular} 
\end{table}

\subsection{Comparison to Other State-of-the-arts}

\begin{figure*}[!t]
\centering \subfigure[]{\label{fig:radiomap_ref_los}\includegraphics[scale=0.315]{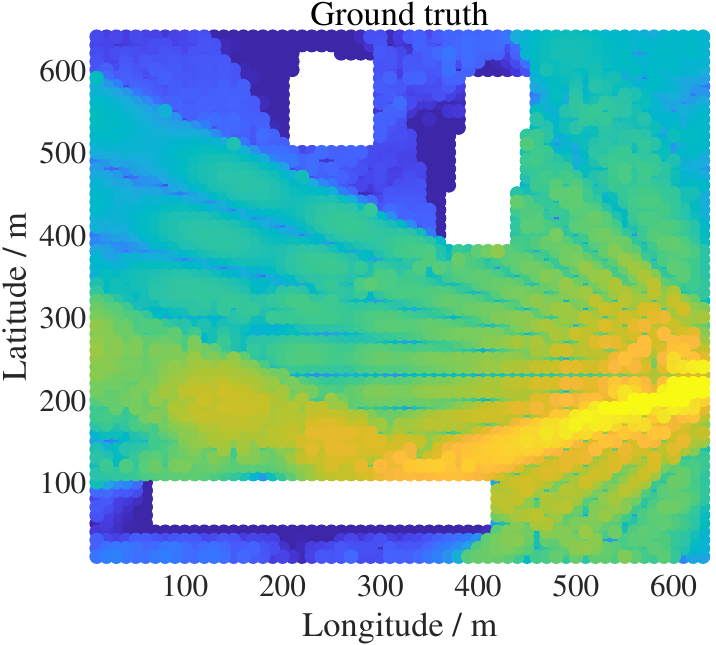}\includegraphics[scale=0.315]{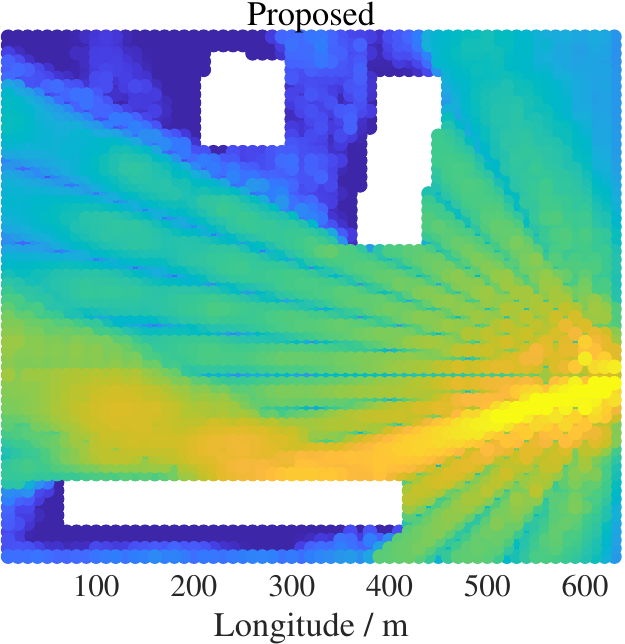}\includegraphics[scale=0.315]{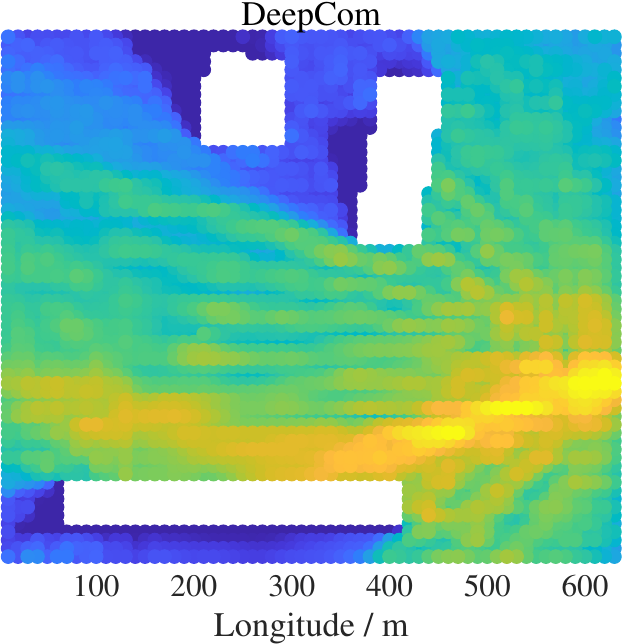}\includegraphics[scale=0.315]{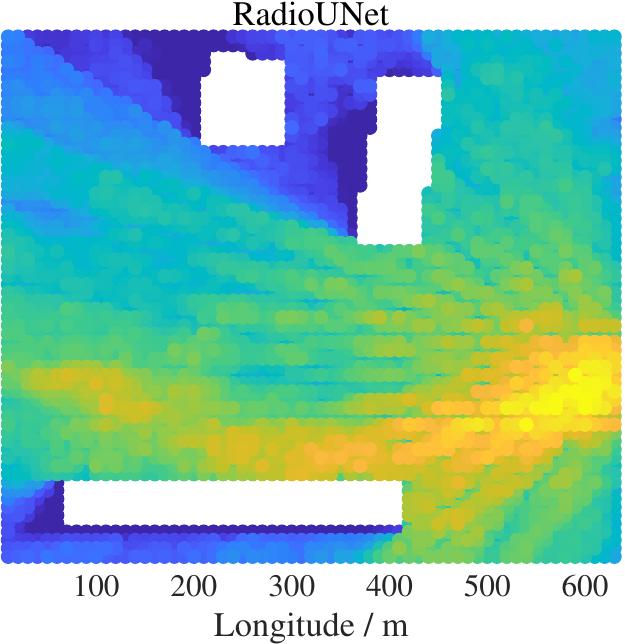}\includegraphics[scale=0.315]{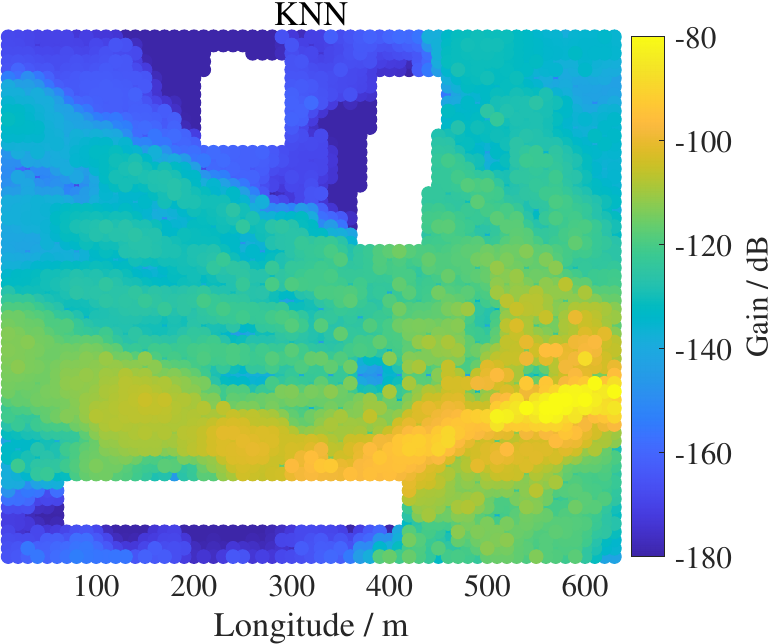}}

\centering \subfigure[]{\label{fig:radiomap_ref_ref}\includegraphics[scale=0.315]{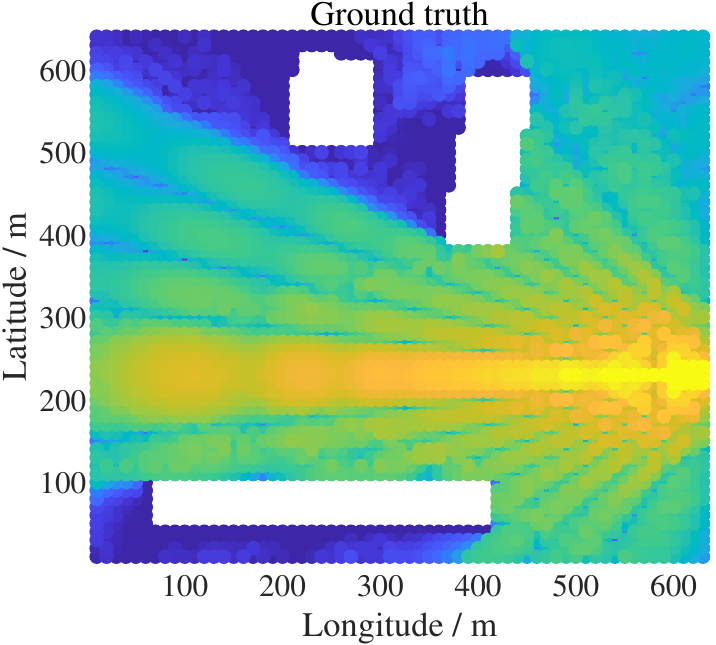}\includegraphics[scale=0.315]{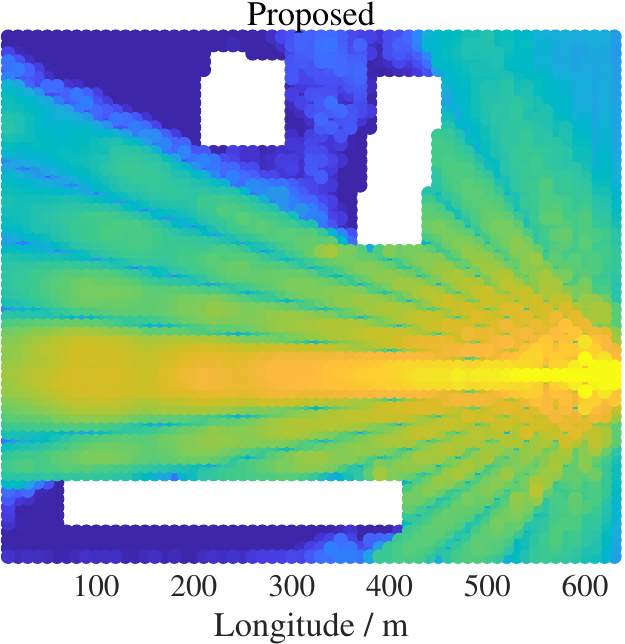}\includegraphics[scale=0.315]{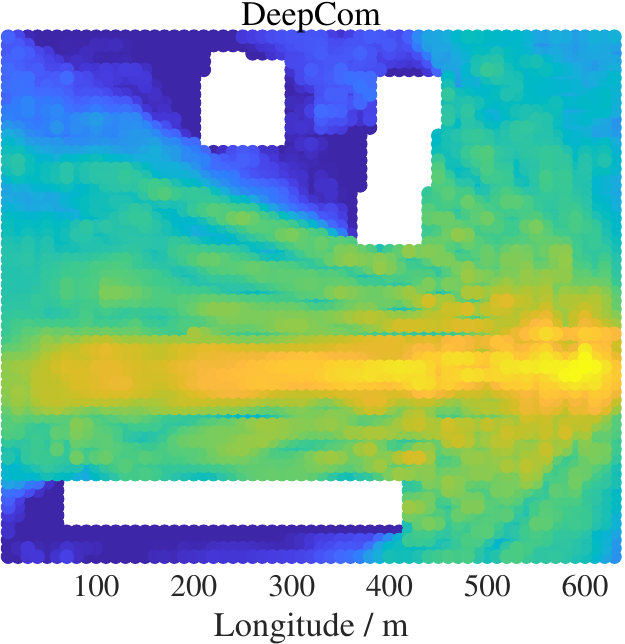}\includegraphics[scale=0.315]{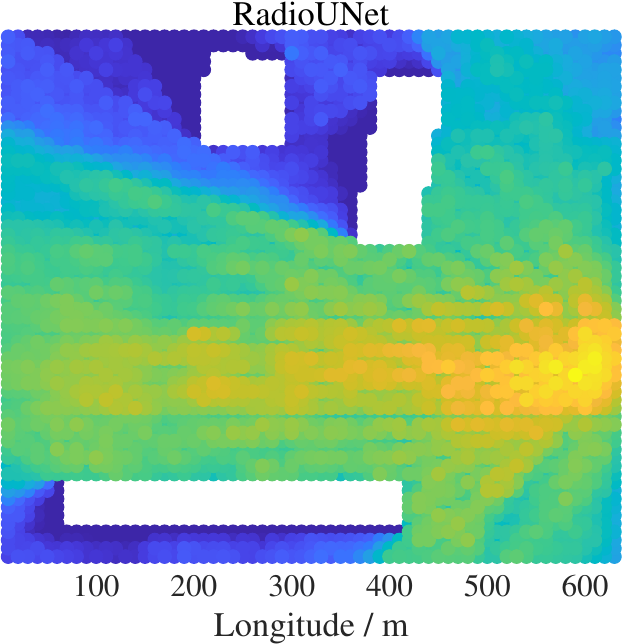}\includegraphics[scale=0.315]{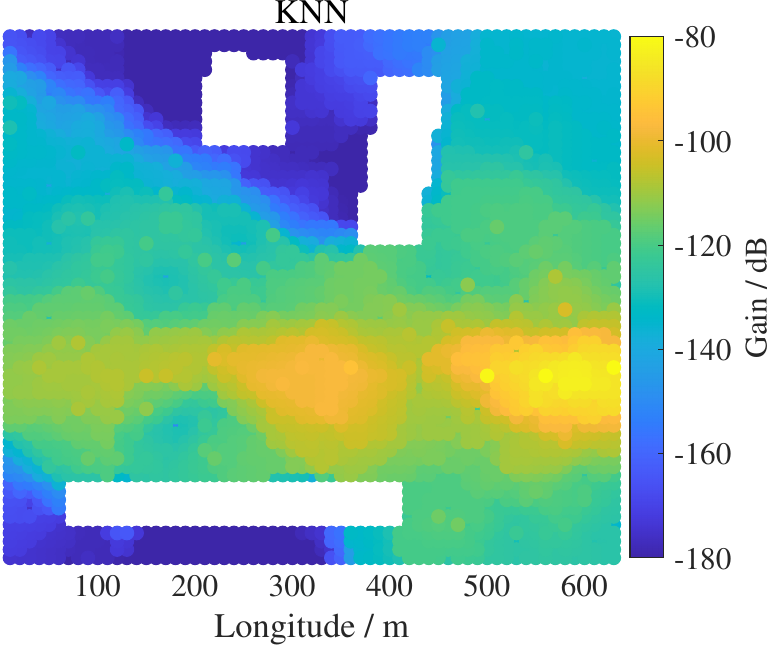}}\caption{The ground truth and reconstructed MIMO beam maps: a) The reflected
beam maps constructed under $30\%$ training samples; b) The direct
beam maps without reflections constructed under $5\%$ training samples.}
\label{fig:RadioMap-Example-Reflection}
\end{figure*}

\begin{figure}[!t]
\hspace{0.2cm}\includegraphics[scale=0.65]{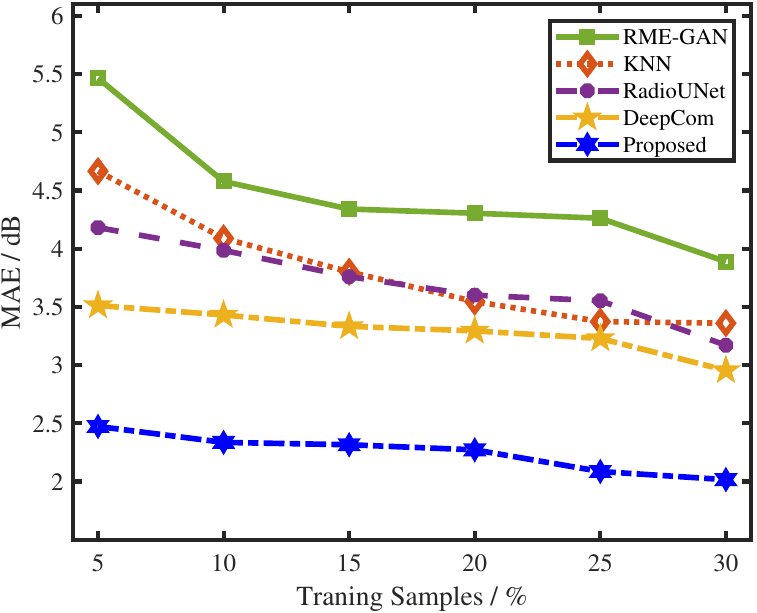}
\caption{Reconstruction MAE versus the ratio of training samples.}
\label{fig:SampleVsMthod}
\end{figure}

We compare the performance of the MIMO beam map construction of the
proposed model and the baselines.

Two slices of the MIMO beam maps are presented in Fig.~\ref{fig:RadioMap-Example-Reflection},
where Fig.~\ref{fig:radiomap_ref_los} illustrates a reflection-involved
scenario, and Fig.~\ref{fig:radiomap_ref_ref} shows a direct beam
case without reflection. The results show that the proposed model
clearly delineates the beam direction and width, as well as both incident
and reflected beams. By contrast, although the baselines can partially
learn beam shapes and reflection with $30\%$ of the training samples,
only DeepCom performs well at $5\%$, as the others fail to capture
beam features. This highlights the advantage of the proposed model
in beam reconstruction across both scenarios.

Table~\ref{Tab:example_com_different_Alg.} presents the MAE and
RMSE of various methods for beam map construction. The proposed model
significantly outperforms all baselines across both metrics. Compared
to mainstream deep learning approaches, such as RadioUNet, RME-GAN,
and DeepCom, it achieves an improvement of $32.2\%-48.1\%$ in MAE
and $24.7\%-39.3\%$ in RMSE. Among the baselines, SVT performs the
worst, likely due to the limited effectiveness of such matrix completion
technique in obstructed environments. In addition, unlike other baselines
that require extra input measurements during inference, the proposed
model does not require any additional data to operate after training,
making it more practical for real-world deployment while maintaining
superior accuracy.

In addition, we evaluate beam map construction across various training
sample ratios, ranging from $5\%$ to $30\%$. Fig.~\ref{fig:SampleVsMthod}
illustrates the relationship between the MAE and the training samples.
It is observed that RME-GAN, KNN, and RadioUNet suffer significant
performance drops as the training samples decrease, whereas the proposed
model and DeepCom degrade slightly. In contrast, the proposed model
exhibited exceptional stability with merely a 0.4 dB variation in
MAE, which shows strong robustness and data efficiency. Notably, the
proposed model also outperforms all baselines in all cases, demonstrating
our superiority in achieving high accuracy even with substantially
reduced data requirements.

\begin{figure*}[!t]
\centering \subfigure[Case I]{\label{fig:Beam Extrapolation case1}\includegraphics[scale=0.315]{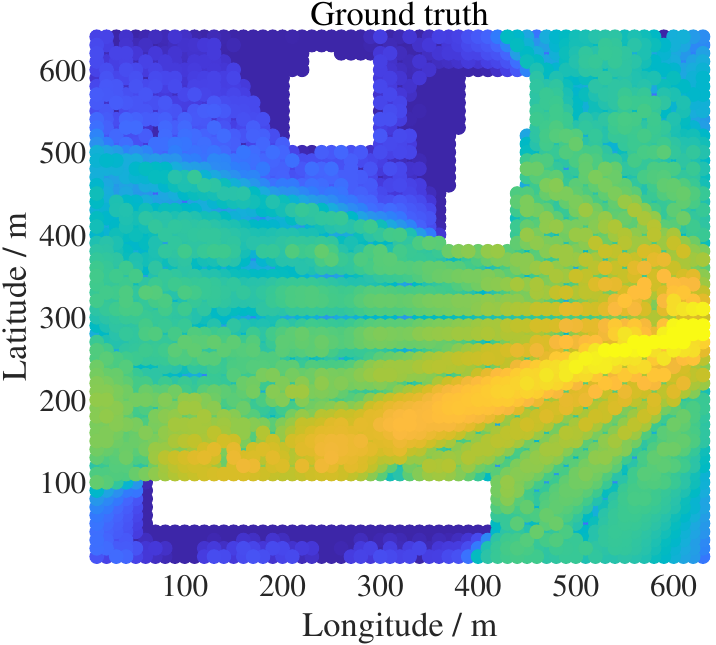}\includegraphics[scale=0.315]{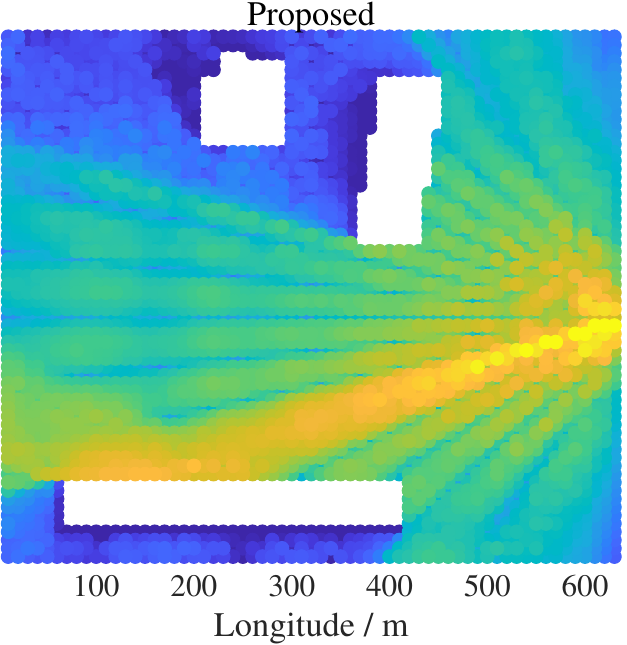}\includegraphics[scale=0.315]{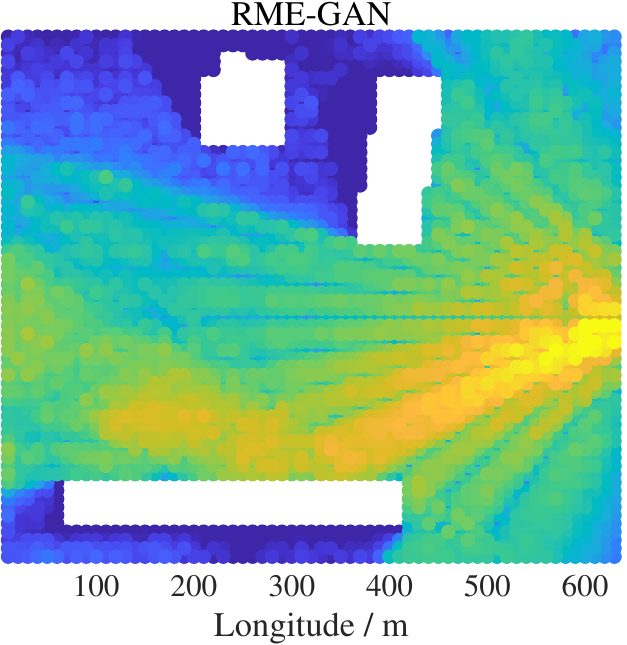}\includegraphics[scale=0.315]{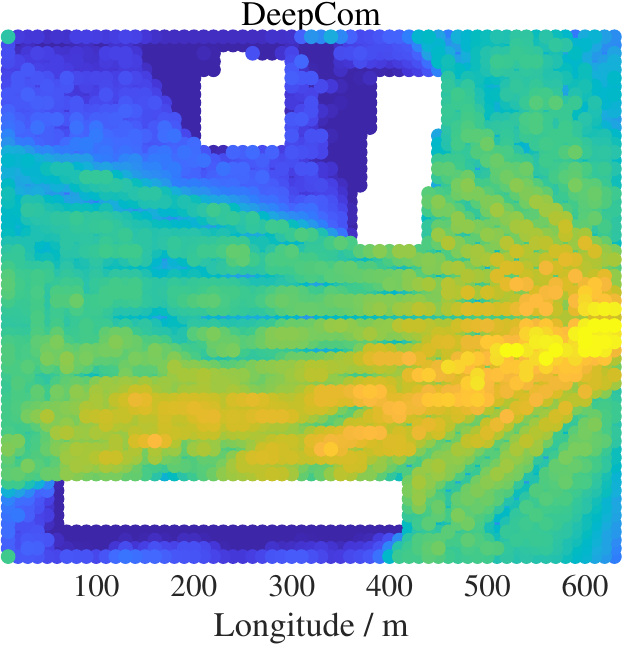}\includegraphics[scale=0.315]{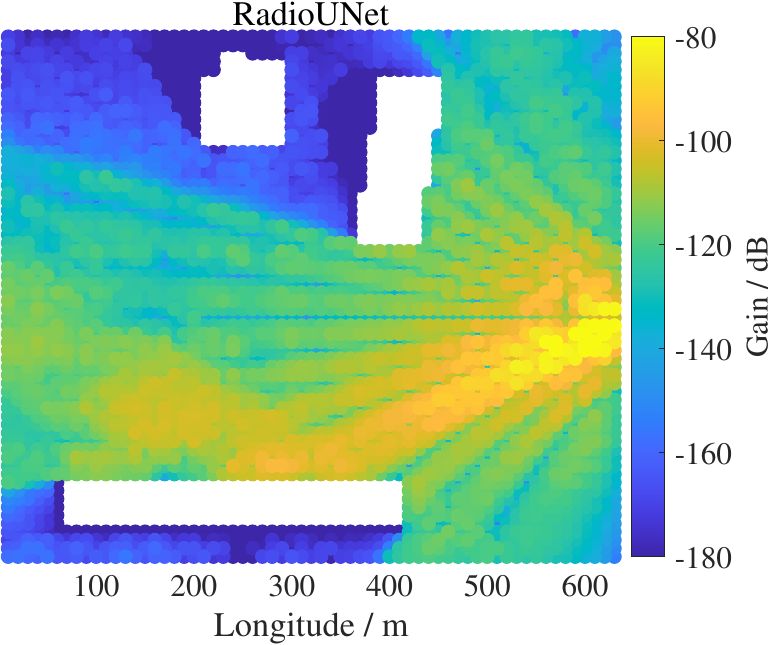}}

{\includegraphics[scale=0.315]{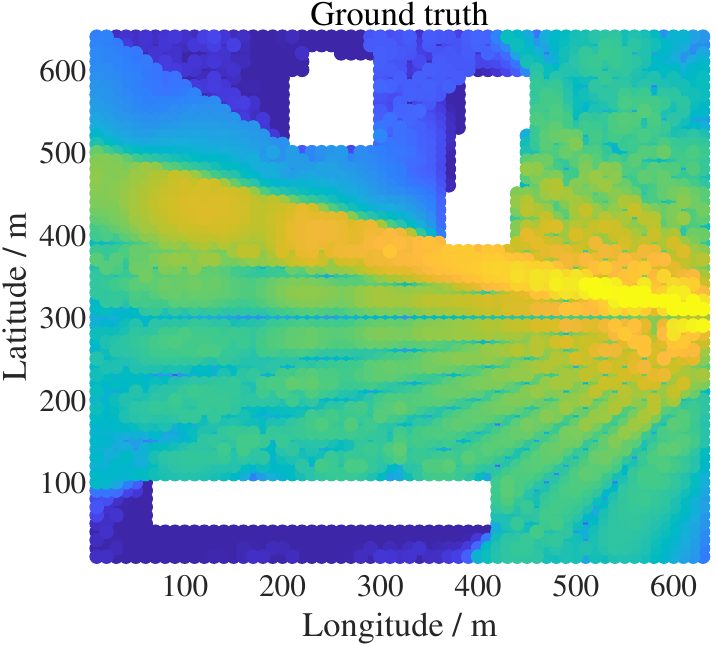}\includegraphics[scale=0.315]{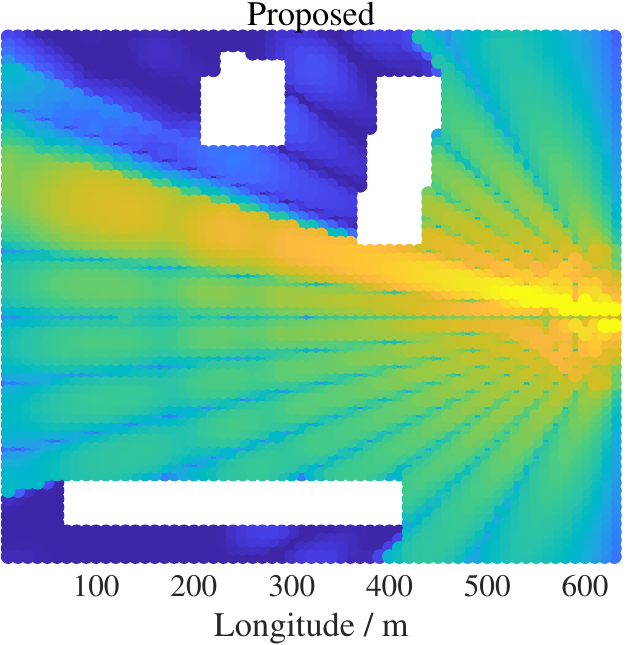}\includegraphics[scale=0.315]{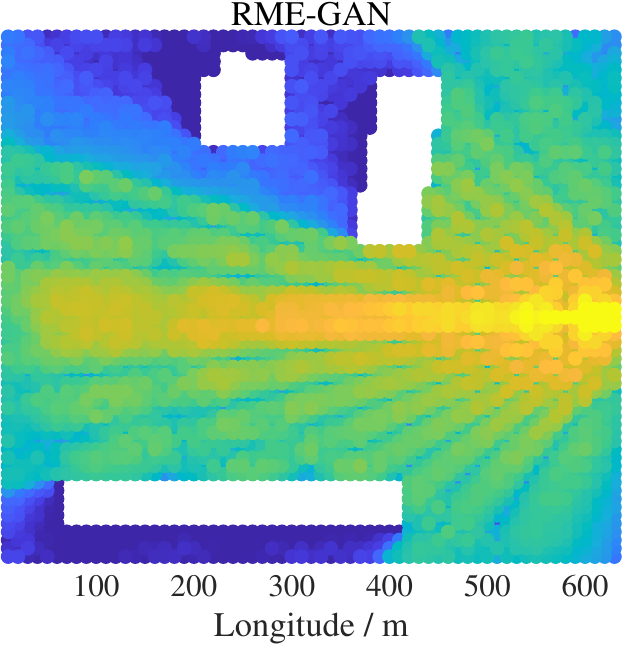}\includegraphics[scale=0.315]{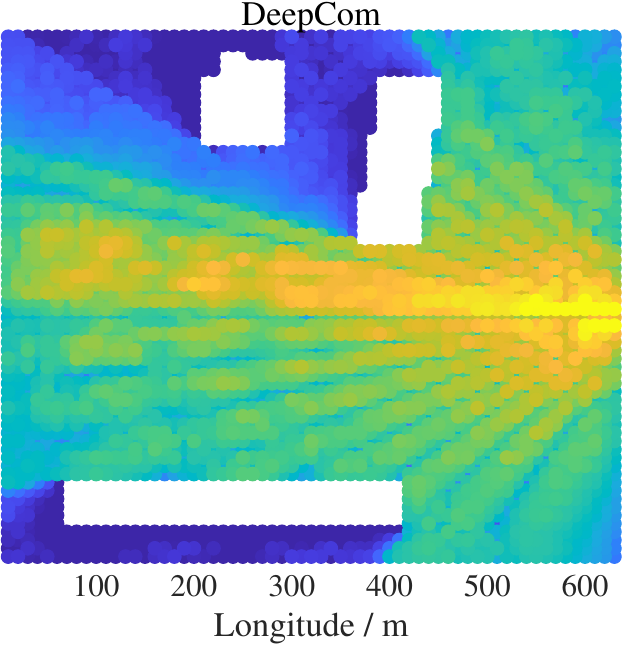}\includegraphics[scale=0.315]{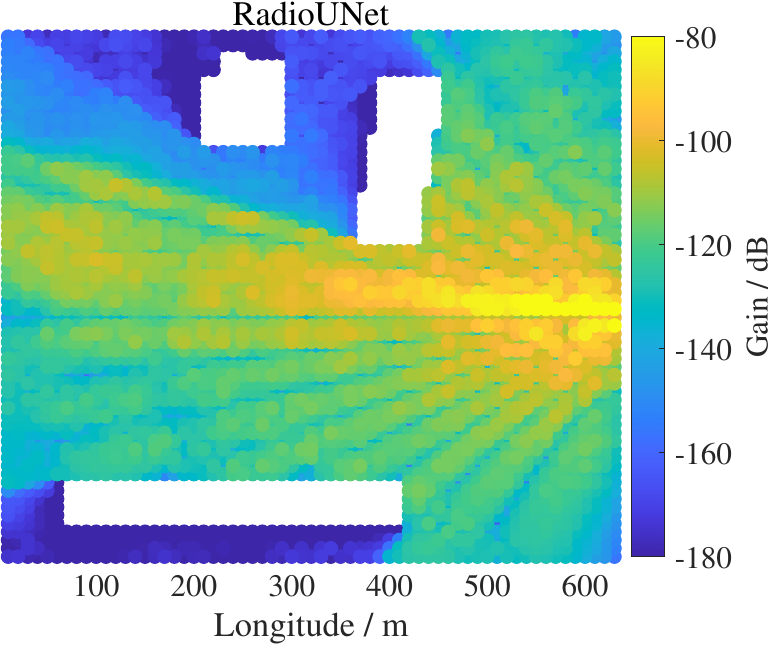}}

\centering \subfigure[Case II]{\label{fig:Beam Extrapolation case 2}\includegraphics[scale=0.315]{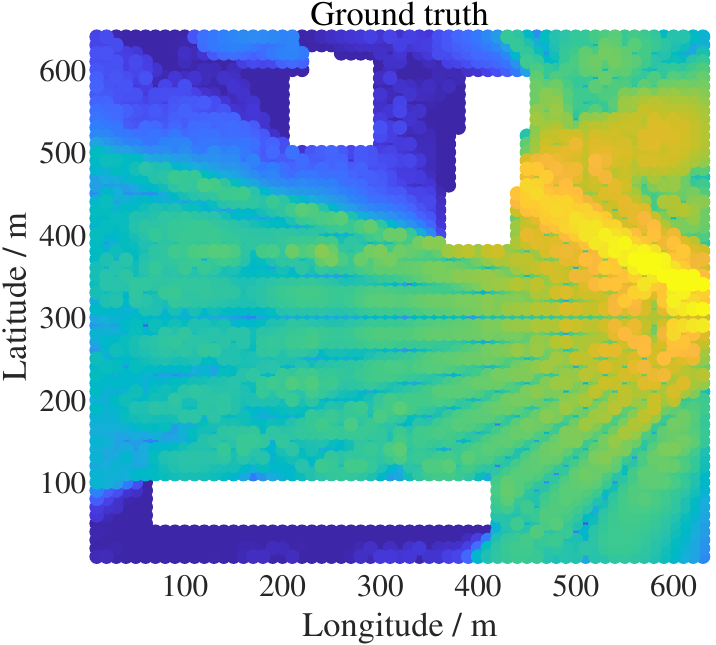}\includegraphics[scale=0.315]{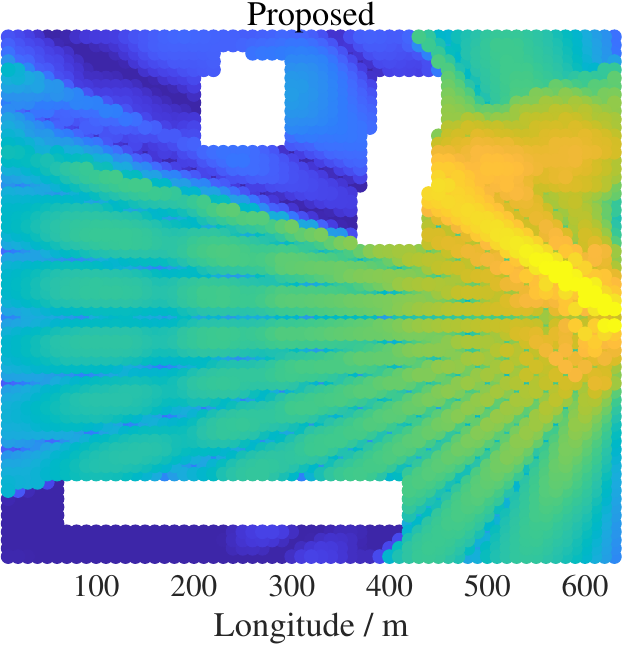}\includegraphics[scale=0.315]{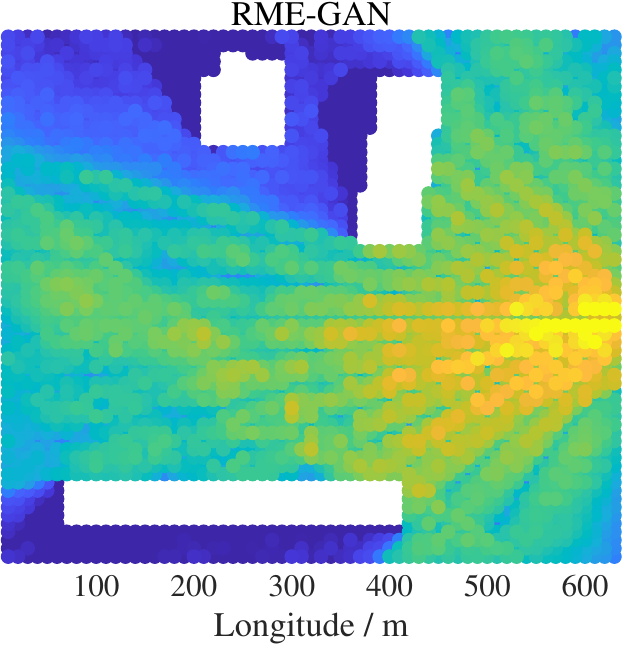}\includegraphics[scale=0.315]{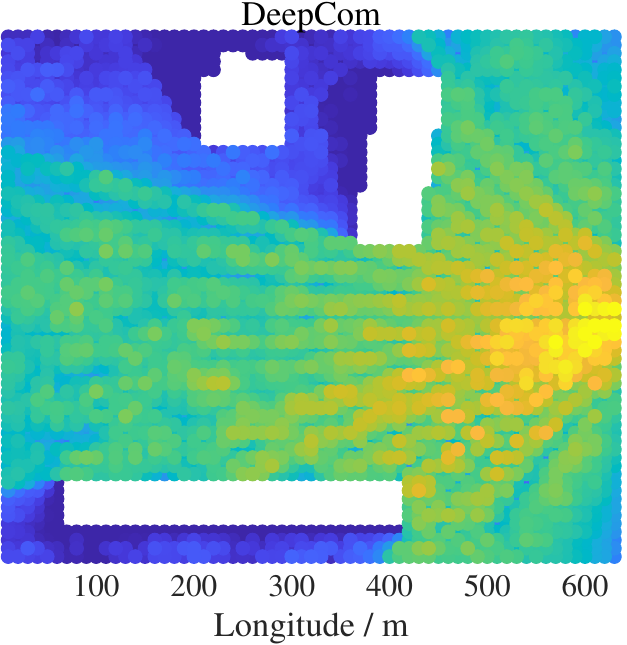}\includegraphics[scale=0.315]{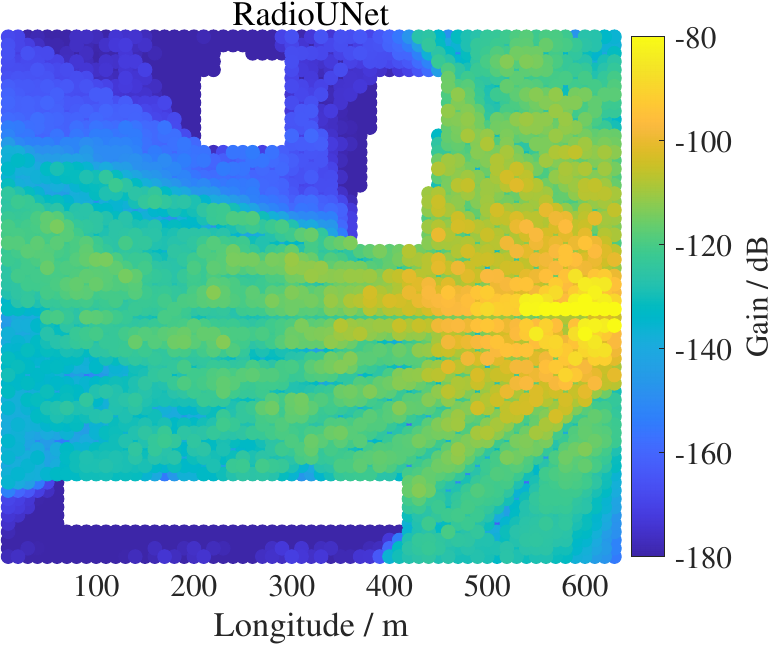}}\caption{MIMO beam extrapolation: a) Trained on odd-numbered beams and tested
on beam No. 6; b) Trained on beams 1-10, and tested on beam No. 11
and 14.}
\label{fig:Example Beam Extrapolation}
\end{figure*}

\subsection{Performance of New Beam Extrapolation}

In this section, we demonstrate that the proposed model can extrapolate
new beams, with its capability to predict beams for previously unseen
scenarios beyond the training data.

To this end, 10 beam maps of a specific \ac{tx} are used for training,
while the remaining maps are reserved for prediction. We consider
two experimental cases: I) training and test beam maps are selected
alternately (e.g., beams 1, 3, 5 for training and 2, 4, 6 for testing);
II) the first 10 beam maps are directly used for training, and the
rest for testing. As a result, in Case~I, adjacent beams in the training
and test sets may exhibit partial spatial overlap, while in Case~II,
such overlap occurs only for the first test beam. We further assume
the environment is known, such that our model can capture reflection
features from unknown obstacles and previously unseen beam directions,
as these can be learned from other \acpl{tx}. Note that in Case~II,
we omit the scattering branch from the proposed model and evaluate
its extrapolation capability only for the direct and reflection branches.

The beam extrapolation results are shown in Fig.~\ref{fig:Example Beam Extrapolation}.
In Case~I, the proposed model can generate previously unseen beams,
with the predictions closely matching the simulated ones, while the
baselines perform poorly. Specifically, DeepCom fails to generate
clear beam patterns, and both RME-GAN and RadioUNet incorrectly predict
the target beam as a different known beam, likely approximating the
unseen beam by selecting the closest match from the training set.
Similar outcomes are observed for the first extrapolated beam in Case~II,
as shown in the first row of Fig.~\ref{fig:Beam Extrapolation case 2}.
For the other beams, all baselines fail to extrapolate completely
unseen beams with no spatial overlap, revealing their limited generalization
beyond the training data. In contrast, the proposed model successfully
extrapolates these beams, demonstrating superior extrapolation capability.

\begin{figure}[!t]
\centering \subfigure[]{\includegraphics[scale=0.32]{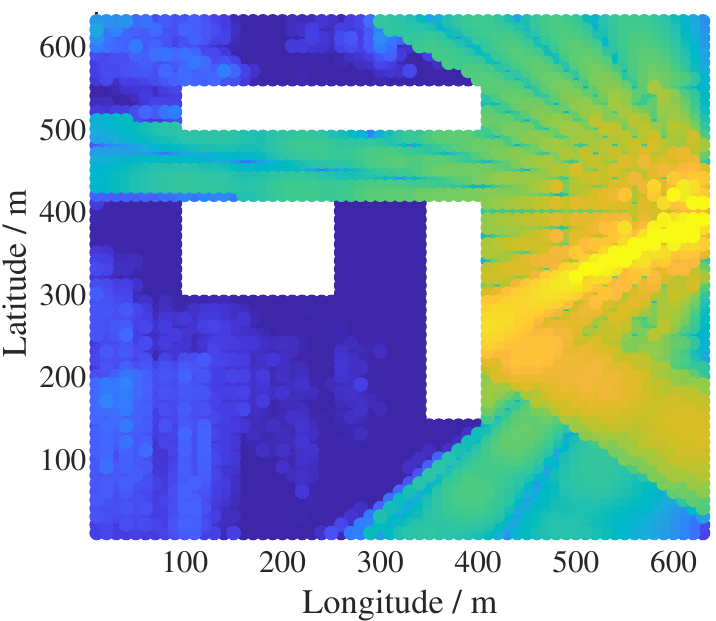}}\hspace{0.5cm}\subfigure[]{\includegraphics[scale=0.32]{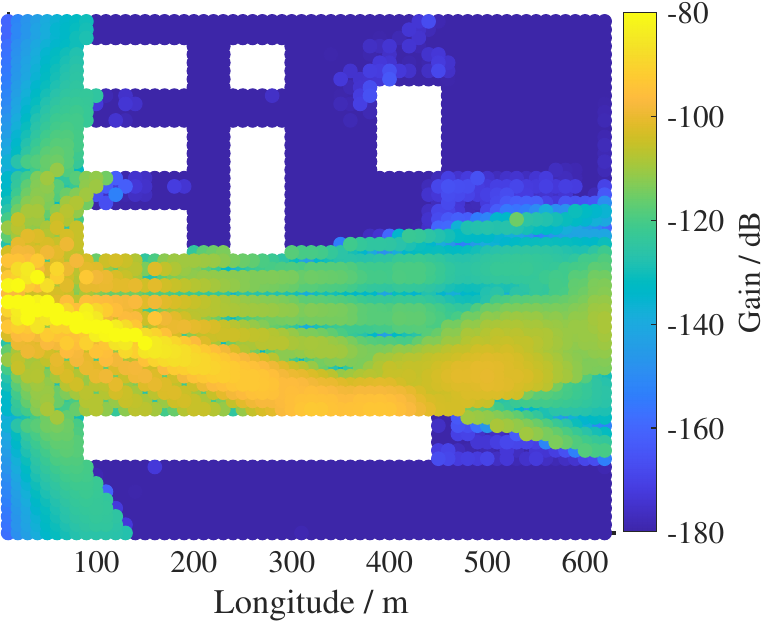}}

\subfigure[]{\includegraphics[scale=0.32]{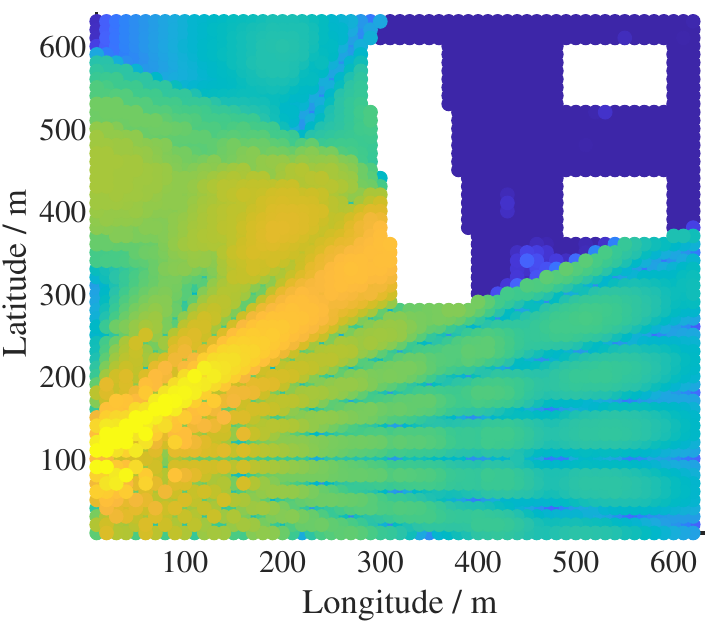}}\hspace{0.5cm}\subfigure[]{\includegraphics[scale=0.32]{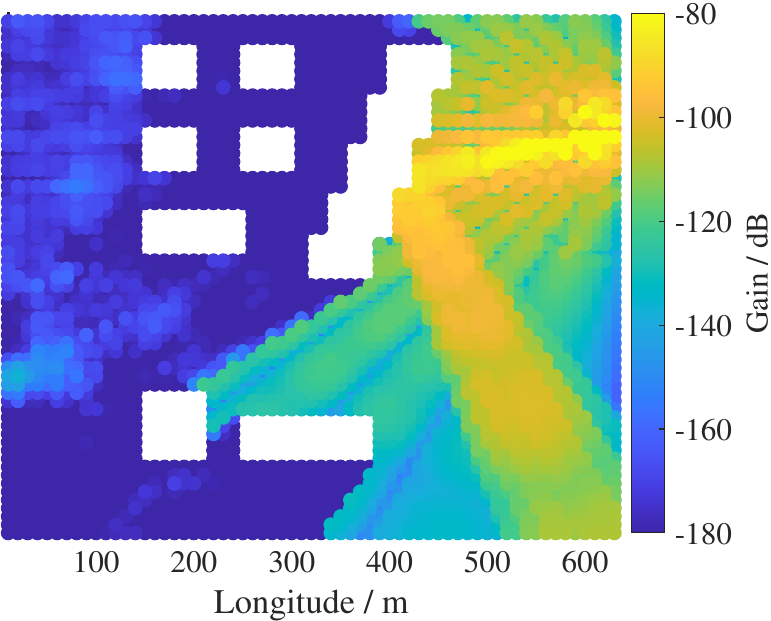}}\caption{MIMO beam maps generated by the proposed model transferred to a new
environment with the knowledge of the oriented virtual obstacles $\mathbf{V}$,
where the trained propagation parameters are fixed.}
\label{fig:RadioMap-Example-NewEnv}
\end{figure}

\subsection{Transferability to a New Environment}

We demonstrate the transferability of the proposed model to a new
environment. The geometric parameters of the new environment will
be input to initialize the virtual environment parameters $\mathbf{V}$,
thereby adapting the model to the new context. 

Fig.~\ref{fig:RadioMap-Example-NewEnv} shows the beam maps generated
by the proposed model across 4 distinct environments with varying
TX positions and obstacle configurations. It is observed that the
proposed model accurately predicts beam patterns in previously unseen
environments based on their geometric parameters, TX locations and
beam angles. The relationship between incident beams, reflections,
and obstacles remains consistent and predictable. Notably, the model
generalizes well across varying TX positions, indicating that it learns
relative spatial features rather than memorizing absolute locations.
Consequently, as long as environment parameters are available or pre-estimated,
the model can adapt to diverse scenarios without retraining, ensuring
efficiency and scalability for deployment.

\subsection{Application in Site-Specific Beam Alignment}

This section discusses the advantages of recognizing blockage and
reflection for online beam alignment.

We consider to select an optimal beamforming vector $\mathbf{w}_{j}$
from a predefined codebook $\mathcal{W}$ at the TX to align the beam
toward a specific RX to maximize the communication quality. The beam
alignment problem is to find the beam index that maximizes the \ac{snr}:
\begin{align}
j^{*}=\mathop{\mbox{argmax}}\limits _{j\in\{1,.,|\mathcal{W}|\}}\frac{|\mathbf{h}_{\mathrm{d}}(\tilde{\mathbf{p}})\mathbf{w}_{j}|^{2}P_{s}}{\sigma^{2}}\label{eq:beam-selection-problem}
\end{align}
where $P_{s}$ denotes the transmitted power and $\sigma^{2}$ is
the noise power. Here $\mathbf{h}_{\mathrm{d}}(\tilde{\mathbf{p}})=\mathbf{h}(\tilde{\mathbf{p}})+\mathbf{h}^{\epsilon}$
represents the dynamic channel, where $\mathbf{h}^{\epsilon}$ captures
uncertainty due to environmental dynamics (e.g., moving scatterers).
In this problem, beam sweeping is employed to identify the optimal
beam index $j^{*}$ from the codebook $\mathcal{W}$.

We propose an environment-aware beam sweeping strategy to reduce search
overhead by leveraging learned environmental geometry. Specifically,
the beam is steered towards the RX in LOS regions, and in NLOS regions,
the search is guided by valid reflected paths, which are summarized
as follows:
\begin{itemize}
\item If $\mathbb{I}\{\tilde{\mathbf{p}}\in\mathcal{\tilde{D}}_{0}(\mathbf{V}_{\textrm{H}})\}$,
then $j^{*}=\textrm{argmin}{}_{j}|\varphi(\mathbf{p}_{\mathrm{t}},\mathbf{p}_{\mathrm{r}})-\theta_{j}|$
and no probing overhead is required.
\item If $\mathbb{I}\{\tilde{\mathbf{p}}\in\mathcal{\tilde{D}}_{1}(\mathbf{V}_{\textrm{H}})\}$,
the beam directions are searched based on valid reflected paths that
satisfy both $\varphi(\mathbf{p}_{\mathrm{t}},\mathbf{c}_{m})\in[\theta_{j}\pm\Theta_{j}/2]$
and $\mathbb{I}\{\mathbf{p}_{\mathrm{r}}\in\tilde{\mathcal{Q}}_{m}(\mathbf{p}_{\mathrm{t}},\mathbf{V})\}$,
as verified by equation (\ref{eq:Ref_Zone_NN}).
\item If $\mathbb{I}\{\tilde{\mathbf{p}}\in\mathcal{\tilde{D}}_{1}(\mathbf{V}_{\textrm{H}})\}$
and no valid reflection exists, all beam directions are searched.
\end{itemize}

The following benchmarks are considered: 1) Exhaustive search: the
best beam is chosen by sweeping all candidates; 2) Hierarchical search
\cite{XiaoHe:J16}: it repeatedly splits the space into two partitions
and refines beam selection from wide to narrow beams; 3) Probing based
method \cite{HenMo:J22}: it jointly learns a compact probing codebook
with channel $\mathbf{h}$ and predict the beam index via a classifier.
We evaluate the beam alignment performance under $P_{s}=30$ dBm and
20 dBm, respectively, with the noise power $\sigma^{2}$ set to -110
dBm.

Table~\ref{Tab:beam alignment} summarizes the average \ac{snr}
across 100 RXs and the number of probing steps required by each method.
Here, the probing step refers to the number of beam directions explored
during the alignment process. It is observed that our map assisted
approach achieves an SNR comparable to that of the exhaustive search
while reducing the number of probing steps by $78\%$. Compared to
hierarchical search, it improves the SNR by 6 dB and requires only
half the number of probing steps. Although the probing based method
reduces the probing steps and achieves higher SNR when $P_{s}=30$
dBm, its performance deteriorates significantly when $P_{s}$ is reduced
to 20\,dBm.

\begin{table}
\caption{Application in Beam Alignment}

\centering\label{Tab:beam alignment} \renewcommand\arraystretch{1.5}
\begin{tabular}{p{2.65cm}<{\centering}|p{0.80cm}<{\centering}|p{1.15cm}<{\centering}|p{0.80cm}<{\centering}|p{1.15cm}<{\centering}}   
\hline   

\multirow{2}{*}{Scheme} & \multicolumn{2}{c|}{$P_{s}$ = 30 dBm} & \multicolumn{2}{c}{$P_{s}$ = 20 dBm} \\  
\cline{2-5}
                         & Steps & SNR (dB) & Steps & SNR (dB)  \\
\hline   
Exhaustive search         & 1600    & \textbf{26.79}   & 1600 & \textbf{15.26}    \\   
\hline   
Hierarchical search       & 800     & 19.10   & 800 & 9.17    \\ 
\hline   
Probing based method      & 400     & 21.33   & 400 & 4.79    \\ 
\hline   
Map assisted (ours)       & \textbf{352}  & \textbf{25.60}   & \textbf{352} & \textbf{14.93}   \\
\hline   
\end{tabular} 
\end{table}

\section{Conclusion}

This paper developed a physics-informed neural network to learn the
geometric features of the virtual environment for environment-aware
MIMO beam map construction. In contrast to many existing deep learning
approaches that lack an explicit environment model, an oriented virtual
obstacle model was proposed to capture the environmental geometry
by modeling signal behaviors of blockage and reflection. To characterize
reflective propagation, we formulated the concept of the reflective
zone, derived its expression, and reformulated it for improved compatibility
with deep learning representations. By integrating this reflective-zone-based
geometry model, a deep neural network architecture was designed to
jointly learn the blockage, reflection and scattering components,
as well as the beam pattern. Design examples showed a great superiority
of the proposed model with over 30\% improvement in beam map accuracy,
along with strong extrapolation to unseen beams and transferability
to new environments.

\begin{appendices}

\section{Proof of Lemma 1}\label{app:Lemma 1 proof}

\begin{figure}[H]
\centering\includegraphics[scale=0.35]{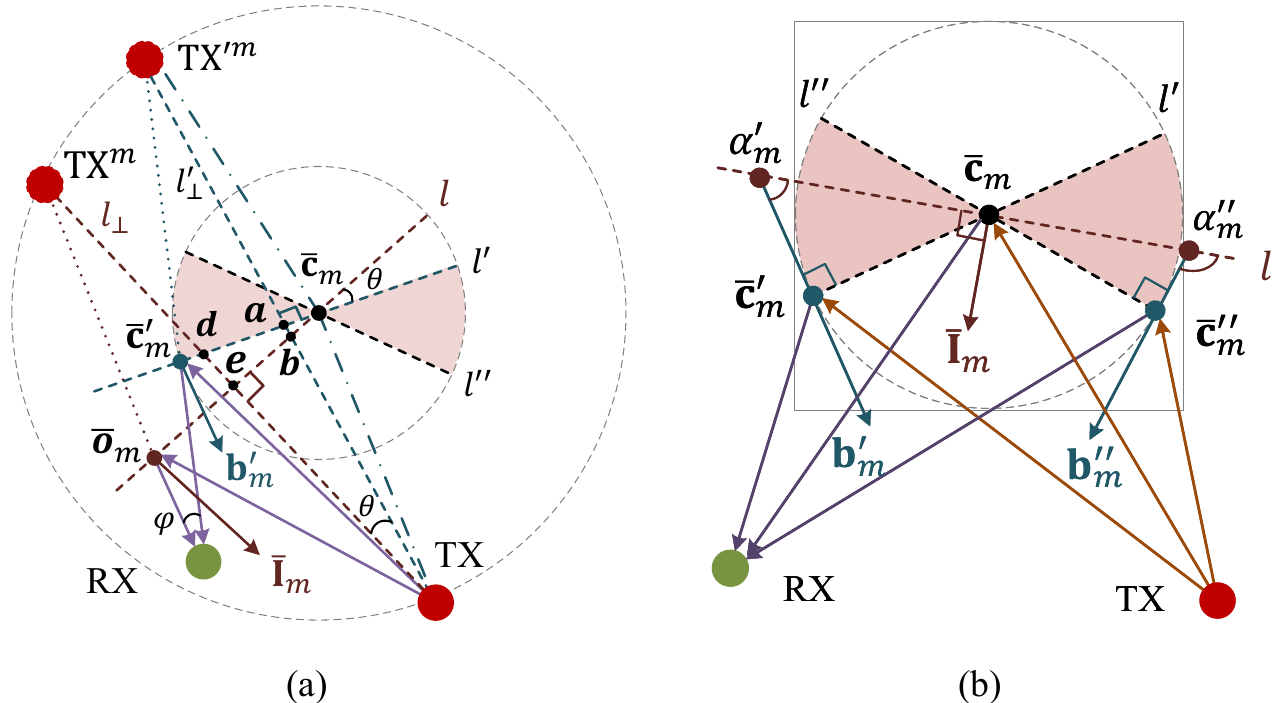}

\caption{The geometry of orientation conditions for the reflective zone.}
\label{fig:Proof}
\end{figure}
The proof of Lemma~\ref{thm:Lemma} proceeds in two steps. First,
we show that if line $l$ does not lie in the plane spanned by $l'$
and $l''$, the point $\bar{\mathbf{o}}_{m}$ falls outside the valid
circular region of $l$, violating the law of reflection. Second,
we prove that if $\bar{\mathbf{o}}_{m}$ lies on $l$ within the circular
region, the lemma holds.

Consider the case where $\bar{\mathbf{o}}_{m}$ lies to the left of
the center $\mathbf{\bar{c}}{}_{m}$. As shown in Fig.~\ref{fig:Proof}(a),
line $l$ is formed by rotating $l'$ counterclockwise by an angle
$\theta$, placing it outside the region bounded by $l'$ and $l''$.
Let TX$^{m}$ and TX$'^{m}$ denote the virtual TXs associated with
$l$ and $l'$, located at $\mathbf{\bar{p}}_{\mathrm{t}}^{m}$ and
$\mathbf{\bar{p}}_{\mathrm{t}}'^{m}$, respectively. Geometrically,
triangles $\triangle\boldsymbol{a}\mathbf{\bar{p}}_{\mathrm{t}}'^{m}\mathbf{\bar{c}}{}_{m}$
and $\triangle\boldsymbol{a}\mathbf{\bar{p}}_{\mathrm{t}}\mathbf{\bar{c}}{}_{m}$
are congruent, implying that $||\mathbf{\bar{p}}_{\mathrm{t}}'^{m}-\mathbf{\bar{c}}{}_{m}||=||\mathbf{\bar{p}}_{\mathrm{t}}-\mathbf{\bar{c}}{}_{m}||$.
Hence, the TX, TX$^{m}$ and TX$'^{m}$ all lie on the same outer
circle centered at $\mathbf{\bar{c}}{}_{m}$ with radius $||\mathbf{\bar{p}}_{\mathrm{t}}-\mathbf{\bar{c}}{}_{m}||$.
As line $l$ rotates, the virtual TX traces a circular arc. In addition,
from the similarity of $\triangle\boldsymbol{a}\boldsymbol{b}\mathbf{\bar{c}}{}_{m}$
and $\triangle\boldsymbol{e}\boldsymbol{b}\mathbf{\bar{p}}_{\mathrm{t}}$,
it follows that $\angle\mathbf{\bar{p}}_{\mathrm{t}}'^{m}\mathbf{\bar{p}}_{\mathrm{t}}\mathbf{\bar{p}}_{\mathrm{t}}^{m}=\theta$.
As $\theta$ increases, TX$^{m}$ moves counterclockwise along the
outer circle, and the angle $\varphi$ between the paths $\overrightarrow{\mathbf{\bar{p}}_{\mathrm{t}}^{m}\mathbf{\bar{p}}_{\mathrm{r}}}$
and $\overrightarrow{\mathbf{\bar{p}}_{\mathrm{t}}'^{m}\mathbf{\bar{p}}_{\mathrm{r}}}$
also increases. Consequently, $\bar{\mathbf{o}}_{m}$ shifts further
from $\mathbf{\bar{c}}{}_{m}$, leading to $||\mathbf{\bar{o}}{}_{m}-\mathbf{\bar{c}}{}_{m}||\geq||\mathbf{\bar{c}}{}_{m}'-\mathbf{\bar{c}}{}_{m}||$,
which implies $\mathbf{\bar{o}}{}_{m}$ lies outside the grid cell,
violating the orientation condition in (\ref{eq:reflective_zone_model}).
Thus, line $l$ must lie within the shadow region spanned by $l'$
and $l''$.

We further verify the conclusion by showing that when a reflected
path exists, the line $l$ lies within the region bounded by $l'$
and $l''$. In Fig.~\ref{fig:Proof}(b), extend the vectors $\mathbf{b}_{m}'$
and $\mathbf{b}_{m}''$ from points $\mathbf{\bar{c}}{}_{m}'$ and
$\mathbf{\bar{c}}{}_{m}''$ in the opposite direction to intersect
line $l$ at points $\mathbf{\boldsymbol{\alpha}}_{m}'$ and $\mathbf{\boldsymbol{\alpha}}_{m}''$,
respectively. By the triangle angle sum property, we obtain $\angle(\mathbf{b}_{m}',\overrightarrow{\mathbf{\boldsymbol{\alpha}}_{m}'\mathbf{\bar{c}}{}_{m}})<\pi/2$
and $\angle(\mathbf{b}_{m}'',\overrightarrow{\mathbf{\bar{c}}{}_{m}\mathbf{\boldsymbol{\alpha}}_{m}''})>\pi/2$.
Given $\angle(\bar{\mathbf{l}}_{m},\overrightarrow{\mathbf{\bar{c}}{}_{m}\mathbf{\boldsymbol{\alpha}}_{m}''})=\pi/2$,
it follows that $\angle(\mathbf{b}_{m}'',\overrightarrow{\mathbf{\bar{c}}{}_{m}\mathbf{\boldsymbol{\alpha}}_{m}''})-$
$\angle(\mathbf{b}_{m}',\overrightarrow{\mathbf{\boldsymbol{\alpha}}_{m}'\mathbf{\bar{c}}{}_{m}})=\angle(\mathbf{b}_{m}'',\overrightarrow{\mathbf{\bar{c}}{}_{m}\mathbf{\boldsymbol{\alpha}}_{m}''})-\angle(\bar{\mathbf{l}}_{m},\overrightarrow{\mathbf{\bar{c}}{}_{m}\mathbf{\boldsymbol{\alpha}}_{m}''})+\angle(\bar{\mathbf{l}}_{m},\overrightarrow{\mathbf{\boldsymbol{\alpha}}_{m}'\mathbf{\bar{c}}{}_{m}})-\angle(\mathbf{b}_{m}',\overrightarrow{\mathbf{\boldsymbol{\alpha}}_{m}'\mathbf{\bar{c}}{}_{m}})$,
which yields the final result $\angle(\mathbf{b}_{m}',\mathbf{b}_{m}'')=\angle(\mathbf{b}_{m}',\bar{\mathbf{l}}_{m})+\angle(\mathbf{b}_{m}'',\bar{\mathbf{l}}_{m})$.

\section{Proof of Proposition 1}\label{app:Proposition 1 proof}

According to Lemma~\ref{thm:Lemma}, the orientation conditions in
the first two terms of (\ref{eq:reflective_zone_model}) can be reformulated
as $\prod_{i=\{1,2\}}\mathbb{I}\{\left|o_{m,i}-c_{m,i}\right|\leq D/2\}\hspace{0.1cm}\mathbb{I}\{\left(\mathbf{l}_{m}\cdot(\mathbf{c}_{m}-\mathbf{p}_{\mathrm{r}})\right)/(\mathbf{l}_{m}\cdot(\mathbf{p}_{\textrm{t}}^{m}-\mathbf{p}_{\mathrm{r}}))\}=\mathbb{I}\{\angle(\mathbf{b}_{m}',\bar{\mathbf{l}}_{m})+\angle(\mathbf{b}_{m}'',\bar{\mathbf{l}}_{m})=\angle(\mathbf{b}_{m}',\mathbf{b}_{m}'')\}$.
Here we omit $\tilde{\mathbf{p}}$ for simplicity. On the right-hand
side, since both $\angle(\mathbf{\widehat{b}}_{m},\mathbf{b}_{m}')$
and $\angle(\mathbf{\widehat{b}}_{m},\bar{\mathbf{l}}_{m})$ are acute,
it follows that $\angle(\mathbf{\widehat{b}}_{m},\bar{\mathbf{l}}_{m})\leq\angle(\mathbf{\widehat{b}}_{m},\mathbf{b}_{m}')$,
and thus, $\cos\angle(\mathbf{\widehat{b}}_{m},\mathbf{l}_{m})\geq\cos\angle(\mathbf{\widehat{b}}_{m},\mathbf{b}_{m}')$.
This is equivalent to $\bar{\mathbf{l}}_{m}\cdot\mathbf{\widehat{b}}_{m}/(||\bar{\mathbf{l}}_{m}||\hspace{0.05cm}||\mathbf{\widehat{b}}_{m}||)\geq\mathbf{b}_{m}'\cdot\mathbf{\widehat{b}}_{m}/(||\mathbf{b}_{m}'||\hspace{0.05cm}||\mathbf{\widehat{b}}_{m}||$.
Since $\mathbf{\widehat{b}}_{m}$ and $\mathbf{b}_{m}'$ are known
a priori, we normalize them to unit vectors, which confirms the first
term. The second term follows directly from (\ref{eq:Height Condition}).

\end{appendices}

\bibliographystyle{IEEEtran}
\bibliography{StringDefinitions,JCgroup,ChenBibCV,Bib/Reference,Bib/RadioMap,Bib/IEEEabrv,Bib/Reference_MIMO}

\end{document}